\documentclass[12pt,a4paper]{article}

\newif\ifJRSSBreview      % double-spaced review format
\JRSSBreviewtrue         % <-- JRSSB submission (double-spaced)
\usepackage[a4paper,margin=1in]{geometry}

\usepackage{setspace}
\ifJRSSBreview
\else
\fi

\usepackage{enumitem}
\setlist[itemize]{topsep=2pt,parsep=0pt,partopsep=0pt,itemsep=1pt,leftmargin=*}
\setlist[enumerate]{topsep=2pt,parsep=0pt,partopsep=0pt,itemsep=1pt,leftmargin=*}

\usepackage{titling}
\usepackage[T1]{fontenc}
\usepackage{lmodern}
\usepackage{microtype}

\usepackage{amsmath,amssymb,amsthm}
\usepackage{mathtools}
\usepackage{bm}
\usepackage{bbm}

\numberwithin{equation}{section}

\usepackage{graphicx}
\usepackage{booktabs}
\usepackage{array}

\usepackage[font=small,labelfont=bf]{caption}
\usepackage[colorlinks=true,linkcolor=blue,citecolor=blue,urlcolor=blue]{hyperref}

\usepackage[round,authoryear]{natbib}
\usepackage{algorithm}
\usepackage{algpseudocode}

\usepackage{etoolbox}
\AtBeginEnvironment{table}{\singlespacing\small}
\AtBeginEnvironment{figure}{\singlespacing\small}
\AtBeginEnvironment{algorithm}{\singlespacing\small}
\AtBeginEnvironment{thebibliography}{\singlespacing\small}

\theoremstyle{plain}
\newtheorem{theorem}{Theorem}[section]
\newtheorem{proposition}[theorem]{Proposition}
\newtheorem{lemma}[theorem]{Lemma}
\newtheorem{corollary}[theorem]{Corollary}

\theoremstyle{definition}

\newtheorem{example}{Example}

\theoremstyle{remark}
\newtheorem{remark}[theorem]{Remark}

\newcommand{\1}{\mathbbm{1}}
\newcommand{\R}{\mathbb{R}}

\newif\ifJRRBSsub
\ifJRSSBreview
  \JRRBSsubtrue
\else
  \JRRBSsubfalse
\fi
\newif\ifIncludeSM
\IncludeSMtrue
\newcommand{\startSupplementaryMaterial}{%
  \clearpage
  \setcounter{page}{1}%
  \pagestyle{plain}%
  \setcounter{section}{0}%
  \setcounter{subsection}{0}%
  \setcounter{subsubsection}{0}%
  \setcounter{equation}{0}%
  \setcounter{theorem}{0}%
  \setcounter{figure}{0}%
  \setcounter{table}{0}%

  \renewcommand{\thesection}{S\arabic{section}}%
  \renewcommand{\thefigure}{S\arabic{figure}}%
  \renewcommand{\thetable}{S\arabic{table}}%
}

\newcommand{\SMfrontmatter}{%
  \clearpage
  \phantomsection
  \begin{center}
  {\Large\bfseries Online Supplementary Material for}\\[4pt]
  {\large\bfseries Goodness-of-Fit Tests for Censored and Truncated Data: Maximum Mean Discrepancy Over Regular Functionals}\\[8pt]
 {\normalsize Juan Carlos Escanciano and Jacobo de U\~na-\'Alvarez}
  \end{center}

  \vspace{0.5em}
  \noindent\textbf{Outline of the Supplementary Material.}
  \begin{enumerate}[label=\textbf{S\arabic*.}, leftmargin=*]
    \item \hyperref[sec:Complete]{Additional results for complete data} 
    \item \hyperref[sec:CS_supp]{Current status data}
    \item \hyperref[SM:sim_random]{Additional numerical results: random sampling (complete data)}
    \item \hyperref[ULLN]{A Uniform Law of Large Numbers}
    \item \hyperref[GeneralProofs]{Proofs of general results}
    \item \hyperref[ExamplesProofs]{Proofs of the examples}
  \end{enumerate}

  \vspace{0.75em}
  \hrule
  \vspace{1.0em}
}

\title{Goodness-of-Fit Tests for Censored and Truncated Data: Maximum Mean Discrepancy Over Regular Functionals}

\author{%
Juan Carlos Escanciano\thanks{Universidad Carlos III de Madrid.} \and
Jacobo de U\~na-\'Alvarez\thanks{Universidade de Vigo.}%
}

\date{February 2026}

\begin{document}
\maketitle
% ===== Abstract + Keywords (compact but readable; main text remains double-spaced) =====
\begingroup
\ifJRRBSsub
  \setstretch{1.15}
\else
  \setstretch{1.05}
\fi
\small
\begin{abstract}
We develop a systematic, omnibus approach to goodness-of-fit testing for parametric distributional models when the variable of interest is only partially observed due to censoring and/or truncation. In many such designs, tests based on the nonparametric maximum likelihood estimator are hindered by nonexistence, computational instability, or convergence rates too slow to support reliable calibration under composite nulls. We avoid these difficulties by constructing a regular (pathwise differentiable) Neyman-orthogonal score process indexed by test functions, and aggregating it over a reproducing kernel Hilbert space ball. This yields a maximum-mean-discrepancy-type supremum statistic with a convenient quadratic-form representation. Critical values are obtained via a multiplier bootstrap that keeps nuisance estimates fixed. We establish asymptotic validity under the null and local alternatives and provide concrete constructions for left-truncated right-censored data, current status data, and random double truncation; in particular, to the best of our knowledge, we give the first omnibus goodness-of-fit test for a parametric family under random double truncation in the composite-hypothesis case. Simulations and an empirical illustration demonstrate size control and power in practically relevant incomplete-data designs.
\end{abstract}

\endgroup

\vspace{0.5\baselineskip}
% (no \newpage)  <-- let Introduction start immediately

%%%%%%%%%%%%%%%%%%%%%%%%%%%%%%%%%%%%%%%%%%%%%%%%%%%%%%%%%%%%%%%%%%%%%%%%%%%%%%
% Main text starts here (double-spaced in JRSSB mode)
%%%%%%%%%%%%%%%%%%%%%%%%%%%%%%%%%%%%%%%%%%%%%%%%%%%%%%%%%%%%
\section{Introduction}\label{sec:intro}

Censoring and truncation are ubiquitous in survival analysis, reliability,
epidemiology, astronomy, and econometrics. In these settings the object of
interest is an event time or duration $X$, but the data provide only partial
information about $X$ through an observation mechanism that depends on
auxiliary variables (censoring and/or truncation times). As a consequence, even
the basic task of assessing whether $X$ follows a postulated parametric model
becomes technically and conceptually nontrivial. Since parametric modeling
continues to be the workhorse assumption in applications, it is therefore
important to develop robust and reliable goodness-of-fit (GOF) tools for
censored and truncated data.

This paper introduces a systematic methodology for developing \emph{omnibus} GOF
tests for the composite hypothesis $\mathcal H_0:\ F \in \{F_\theta:\theta\in\Theta\}$, where $F$ denotes the distribution function of $X$ and the observed data arise
under censoring and/or truncation. Our goal is a procedure that is (a) valid
under composite nulls, (b) tractable under complex incomplete-data mechanisms,
and (c) does not rely on (non-regular) nonparametric maximum likelihood estimation (NPMLE) of $F$ as an
intermediate object.

A natural and classical route in the GOF literature is to compare a nonparametric estimator
$F_n$ (e.g., a Kaplan--Meier type estimator adapted to the design) to the null
model $F_{\hat\theta}$ using Kolmogorov--Smirnov or Cram\'er--von Mises type
distances. In incomplete-data designs, however, this NPMLE-based route may be
problematic: the NPMLE may fail to exist or be non-unique in relevant models
(see, e.g., \cite{XH2019}); even when it exists and is unique, it may converge
at slow rates (see, e.g., \cite{GroeneboomWellner1992}), or fail to converge weakly without
additional conditions (see, e.g., \cite{SohlTrabs2012}). These issues are particularly acute
under composite nulls, where bootstrapping an NPMLE-based statistic typically
requires repeated re-estimation inside each resample, often under delicate
tail and regularity conditions (see the discussion before Section \ref{subsec:double_trunc}) .

\subsection{Motivation through two canonical designs}\label{subsec:two_schemes}

We motivate both the statistical problem and the structure of our solution
through two representative observation schemes that recur throughout the paper.

\paragraph{Left truncation with right censoring (LTRC).}
In cohort and registry studies individuals enter observation after a truncation
time $U$ and are then followed until the event time $X$ or a censoring time $C$.
The observed outcome is $Y=\min\{X,C\}$ with censoring indicator
$\delta=\1\{X\le C\}$, and one observes $Z=(Y,U,\delta)$ conditional on $U\le Y$.
Even when the model for $X$ is parametric, the distribution of $Z$ depends on
nuisance features of the observation mechanism (the law of $(U,C)$), and this
dependence contaminates the null distribution of naive GOF statistics unless
the procedure is designed to be insensitive to such nuisances.

\paragraph{Random double truncation.}
In double truncation the event time $X$ is observed only if it falls within a
random window $[U,V]$, i.e., one observes $Z=(X,U,V)$ conditional on
$U\le X\le V$. This design arises in epidemiology (delayed sampling combined
with administrative constraints) and in astronomy (magnitude-limited surveys).
Double truncation illustrates a second, deeper difficulty: even defining a
computationally stable, omnibus GOF statistic under a \emph{composite} null is
challenging because the observation mechanism induces strong selection bias and
because the standard nonparametric objects (e.g., Efron--Petrosian type NPMLEs)
may be unstable and hard to bootstrap in practice. This paper places special
emphasis on random double truncation, and, to the best of our knowledge,
provides the first omnibus GOF test for a parametric family under random double
truncation in the composite-hypothesis case (Section~\ref{subsec:double_trunc}
and Section~\ref{sec:num}).

These two schemes highlight the central point: the mapping from the target cumulative distribution function (cdf)
$F$ to the observable distribution is indirect and involves nuisance features
(typically infinite-dimensional) with a potentially non-regular inverse. As a result, classical GOF tests for complete
data do not transfer directly, and NPMLE-based transfers can be fragile or
computationally burdensome.

\subsection{Methodological overview}\label{subsec:overview}

Our approach is based on two key ingredients: (i) a \emph{regular functional}
characterization of the null via a score process, and (ii) a \emph{reproducing kernel Hilbert space} (RKHS)
aggregation that yields a maximum mean discrepancy (MMD) type quadratic form.

\paragraph{Regular moments and an orthogonal score process.}
Rather than comparing $F$ to $F_{\theta_0}$ through a nonparametric estimator of
$F$, where $\theta_0 \in \Theta$ denotes the
unknown parameter such that $F = F_{\theta_0}$ under $\mathcal{H}_0$, we characterize $\mathcal H_0$ by a family of moment restrictions indexed
by test functions $\varphi$:
\[
\mathcal H_0 \Longleftrightarrow E\!\left[g_\varphi(Z_1;\theta_0, G)\right]=0
\quad\text{for all }\varphi \text{ in a rich class}.
\]
Here $Z_1$ denotes the observed data and $G$ an additional nuisance function. For each observation scheme, the score $g_\varphi$ is chosen so that
it is unbiased under $\mathcal H_0$, its expectation is regularly estimable, and rich enough to identify violations of
the null on the identifiable region of the support.

To handle composite nulls and nuisance estimation, we then construct a
\emph{Neyman-orthogonal} (locally insensitive) version of the moment, denoted
$\Pi^\perp\varphi(Z)$, based on the IF representation of the
plug-in sample mean process $E_n[g_\varphi(Z_1;\hat\theta,\hat{G})]$ (cf. Assumption A1), where henceforth $E_n[f] := E_n[f(Z_1)] := \frac{1}{n}\sum_{i=1}^n f(Z_i)$. Orthogonality is the key device
that allows calibration by a multiplier bootstrap \emph{without re-estimating}
nuisance objects inside each resample. This is particularly valuable in
truncation settings where NPMLE re-fitting can be computationally expensive and
theoretically delicate.

\paragraph{RKHS aggregation and induced kernels.}
To obtain an omnibus test, we aggregate the orthogonal score process over a
rich function class by embedding $\varphi$ in a RKHS. Specifically, letting $\mathcal B_K$ be the unit ball of an RKHS
$\mathcal H_K$ with kernel $K$, we define
\[
\hat Q_K=\sup_{\varphi\in\mathcal B_K}\left(E_n[\hat\Pi^\perp\varphi(Z)]\right)^2
\]
where $\hat\Pi^\perp$ stands for an estimator of $\Pi^\perp$. A central practical feature is that $\hat Q_K$ and its bootstrap analog $\hat{Q}_{K}^{\ast}$ admit simple quadratic-form representations 
\begin{equation}\label{eq:QK_quadratic_form}
\hat{Q}_{K}
=
\frac{1}{n^{2}}\sum_{i=1}^{n}\sum_{j=1}^{n}\hat{K}_{g}^{\perp}(Z_{i},Z_{j}),
\qquad
\hat{Q}_{K}^{\ast}
=
\frac{1}{n^{2}}\sum_{i=1}^{n}\sum_{j=1}^{n}\omega_{i}\omega_{j}\,\hat{K}_{g}^{\perp}(Z_{i},Z_{j}),
\end{equation}
\noindent with an \emph{induced} kernel $\hat K_g^\perp$ given in Lemma \ref{Lemma:InducedKernel} and bootstrap multipliers $(\omega_{i})$, yielding
computationally straightforward $V$-statistics. This representation is what makes
the method scalable and easy to implement.

\subsection{Contributions}\label{subsec:contrib}

The main contributions of the paper are as follows.

\begin{enumerate}[leftmargin=*,itemsep=2pt]
\item \textbf{A general GOF testing template for incomplete data via regular functionals.}
We provide a systematic construction of score processes that characterize
$\mathcal H_0$ through regular (pathwise differentiable) functionals, and we
orthogonalize these processes to handle composite null hypotheses under
censoring and truncation.

\item \textbf{RKHS/MMD aggregation with a quadratic-form statistic.}
Aggregating the orthogonal score process over an RKHS unit ball yields an
omnibus discrepancy with a convenient induced-kernel representation, leading to
the simple quadratic-form $V$-statistic and a multiplier bootstrap in (\ref{eq:QK_quadratic_form}).

\item \textbf{Composite-null omnibus GOF testing under random double truncation.}
To the best of our knowledge, ours is the first omnibus GOF test for a parametric
family under random double truncation in the composite-null case.

\item \textbf{Composite-null GOF testing under complete data.}
Even for complete data, our results are new and improve upon existing results in, e.g., \cite{Lindsay2008,Lindsay2014} and \cite{key2021composite}.

\item \textbf{Asymptotic theory under the null and local alternatives.}
We establish the asymptotic null distribution of $n\hat Q_K$, provide a novel Uniform Law of Large Numbers (ULLN), and use it to establish the non-trivial local power properties against Pitman-type alternatives.
\end{enumerate}

\subsection{Relation to existing work}\label{subsec:related}

Our work connects three strands of literature. The first concerns inference and GOF testing under censoring and truncation, where classical procedures
typically compare nonparametric estimators of the target distribution (e.g., Kaplan--Meier,
Nelson--Aalen, Lynden--Bell, or Efron--Petrosian type estimators) to a parametric
counterpart (see, e.g., \citealp{Koziol1980,Nair1981,FlemingHarrington1991,KleinMoeschberger2003,Emura2015}).
In complex incomplete-data designs, these routes may be hindered by nonexistence or instability
of the NPMLE, slow rates, or by boundary/integrability conditions required for weak convergence
on the relevant support (see, e.g., \citealp{Gill1983,SohlTrabs2012,Stute1995, Nikabadze1997,Sun1997,KoulYi2006}).

To be specific, Kolmogorov-Smirnov and Cramér-von Mises type tests for censored data were investigated by \citealp{KoziolGreen1976, Koziol1980, Fleming1980} and \cite{Hjort1990}, among others. Some of these omnibus methods were compared in simulated scenarios by \cite{Balakrishnan2015}, and implemented in software packages for their friendly application (\citep{Besalu2025}). For left-truncated data, Kolmogorov-Smirnov-type tests were considered in \cite{Guilbaud1988} and \cite{Sun1997}, while \cite{Emura2015} proposed bootstrap versions of Kolmogorov-Smirnov and Cramér-von Mises type tests under double truncation for the simple null case. All these methods are NPMLE-based and suffer from important limitations. In particular, from a theoretical viewpoint they rely on weak convergence results on bounded intervals to avoid problems with the potentially erratic behaviour of the NPMLE near the boundary. Importantly, for an extension of the weak convergence results to the whole support of the observable event time, the NPMLE-based approach requires integrability conditions involving the censoring or truncation distributions which may fail in practical settings. See Section \ref{sec:examples} for further details. Our approach avoids all these difficulties because it is NPMLE-free. For current status data (interval censoring case 1), our test is related to the Cram\'{e}r-von Mises test in \cite{KoulYi2006}; see the Supplementary Material. However, the approach for the asymptotic analysis in \cite{KoulYi2006} is different to ours, relying on weak convergence on compacta using \cite{Stute1998}. Again, our RKHS framework allows for convergence results on arbitrary supports. For interval censoring case 2,~\cite{Ren2003,Omidi2021} proposed NPMLE-based tests with slow convergence. We are not aware of other formal goodness-of-fit tests for interval-censored data. 

The second strand concerns discrepancy-based testing in a RKHS, including MMD statistics (see, e.g.,
\citealp{Gretton2012,Sriperumbudur2011}). While RKHS-based GOF tests are well
developed for complete data and simple hypotheses, existing results under \emph{composite} nulls are more limited (see, e.g., \citealp{Lindsay2008,Lindsay2014,key2021composite}). In particular, our test provides an alternative approach to that of \cite{Lindsay2008} for complete data, which is based on Satterthwaite approximations and non-projected quadratic distances. \cite{key2021composite} consider MMD tests for a composite hypothesis, but they restrict the analysis to a specific minimum distance estimator and a parametric bootstrap that requires re-estimation in each bootstrap sample. 

The third strand is the semiparametric paradigm of orthogonal scores and locally robust inference for GOF problems (see, e.g., \citealp{Neyman1959,Bickel1993,Escanciano2024}). In particular,
\citet{Bickel2006} and \cite{Escanciano2024} studied Neyman-score processes, but they do not consider censored or truncated
sampling.

%Our contribution is to focus on regular functionals and to bring the Neyman-orthogonal score-process viewpoint into incomplete-data GOF testing and to couple it with RKHS aggregation. This yields an omnibus MMD-type statistic with a multiplier bootstrap that keeps nuisance estimates fixed. In particular, for random double truncation we provide the first omnibus GOF test for a parametric family under a composite null.

\subsection{Organization}\label{subsec:org}

Section~\ref{sec:setup} introduces the general incomplete-data framework and the
orthogonal score process. Section~\ref{Sec:Asymptotics} presents the general
asymptotic results and the quadratic-form representation used for computation.
Section~\ref{sec:examples} specializes the construction to left truncation with
right censoring and random double truncation (the Supplementary Material
covers current status data). Section~\ref{sec:num}
reports Monte Carlo evidence and an empirical illustration, and
Section~\ref{sec:disc} concludes. Proofs, novel asymptotic results, and additional simulations are gathered in the Supplementary Material.

%%%%%%%%%%%%%%%%%%%%%%%%%%%%%%%%%%%%%%%%%%%%%%%%%%%%%%%%%%%%%%%%%%%%%%%%%%%%%%
\section{Setup and orthogonal score process}
\label{sec:setup}

\subsection{A general incomplete-data model}

Let $X$ denote a target random variable with cdf $F$. Let $W$ be an auxiliary vector of
censoring and truncation random variables with cdf $G$, independent of $X$. The variables
$(X,W)$ are linked to the observed data through a known measurable mapping $Z=T(X,W)$. In addition, $Z$ is observed only when $B(X,W)=1$, where $B$ is also known. Thus, we observe an iid sample $\{Z_i\}_{i=1}^n$ from
$P\equiv P_{F,G}$, the conditional distribution of $Z=T(X,W)$ given $B(X,W)=1$.
This framework covers a broad class of censoring and truncation schemes through suitable
specifications of $T$ and $B$ (see Section~\ref{sec:examples}).

For any measurable $f$ with $E[f^2(Z_1)]<\infty$, write $Pf:=E[f(Z_1)]$, where $E$ denotes expectation under $P=P_{F,G}$. To emphasize null-based expectations and
probabilities we write $E_{\theta,G}$ and $P_{\theta,G}$, or simply $E_\theta$ and $P_\theta$ when dependence on $G$ does not affect the discussion. Let $L_2(P)$ denote the set of
measurable $f$ with $\|f\|_2^2:=E[f^2(Z_1)]<\infty$.

\subsection{Regular moments, scores, and characterization of the null}
\label{sec:full_cond_scores}

%We test the composite null $\mathcal H_0:\ F\in\{F_\theta:\theta\in\Theta\}$, where $\theta_0$ denotes the (pseudo-)true value under $\mathcal H_0$.

\paragraph{Full-score moments.}
For a measurable $\psi\in L_2(P)$, define the full score
\begin{equation}\label{eq:fullscore}
g^{\mathrm{full}}_\psi(Z_1,\theta,G)
\;:=\;
\psi(Z_1)\;-\;E_{\theta,G}[\psi(Z_1)] .
\end{equation}
Then $E[g^{\mathrm{full}}_\psi(Z_1,\theta_0,G)]=0$ under $\mathcal H_0$. This is general but
typically cumbersome in censoring/truncation designs because $E_{\theta,G}[\psi(Z_1)]$
depends on the nuisance $G$.

\paragraph{Conditional-score moments and motivation.}
To obtain implementable omnibus restrictions when $W_1$ is observable, we use conditional moments given the observation variables. Let $\psi\in L_2(P)$, and define
\begin{equation}\label{eq:condscore_general}
g^{\mathrm{cond}}_\psi(Z_1,\theta)
\;:=\;
\psi(Z_1)\;-\;E_{\theta}\!\left[\psi(Z_1)\mid W_1\right].
\end{equation}
Under $\mathcal H_0$, $E[g^{\mathrm{cond}}_\psi(Z_1,\theta_0)]=0$.
In the leading censoring/truncation designs of Section~\ref{sec:examples}, the conditional
law of the target variable given $(W,B=1)$ is identified by $F_\theta$ and $B$ and does not
involve unknown features of $G$, so \eqref{eq:condscore_general} yields nuisance-robust
moments.

We index scores by test functions $\varphi$ through a choice $\psi=\psi_\varphi$ (e.g.,
$\psi_\varphi(Z_1)=\varphi(X_1)$, or $\psi_\varphi(Z_1)=\varphi(Y_1)\delta_1$ in LTRC), and set
\begin{equation}\label{eq:gphi_cond}
g_\varphi(Z_1,\theta)
\;:=\;
g^{\mathrm{cond}}_{\psi_\varphi}(Z_1,\theta)
\;=\;
\psi_\varphi(Z_1)\;-\;E_{\theta}\!\left[\psi_\varphi(Z_1)\mid W_1^o\right]
\end{equation}
where $W_1^o$ stands for the observable components of $W_1$. We follow this strategy in the settings of truncated data and current status data. With censored data, however, the conditional expectation in \eqref{eq:gphi_cond} will depend on the censoring variable which does not belong to $W_1^o$. In such a case, namely the LTRC model, the terms involving the unknowns are replaced by quantities that only depend on components of $Z_1$ and that are valid in the sense of having the same conditional mean. See Section~\ref{sec:examples} where we give explicit formulas for the conditional score in
\eqref{eq:gphi_cond} for random sampling, LTRC, and random double truncation (see the Supplementary Material for current status data).

\paragraph{Null characterization.}
Our theoretical framework allows for complex score functions $g_\varphi$ depending on the cdf $G$ of the truncating/censoring variables. Let $\mathcal C$ be a rich class of test functions (in our applications, indicator classes
and RKHS unit balls). Under $\mathcal H_0$,
\begin{equation}\label{eq:null_moments}
E[g_\varphi(Z_1,\theta_0,G)]=0 \quad \text{for all }\varphi\in\mathcal C.
\end{equation}
Converses (on the identifiable region of the support) are verified case-by-case under minimal
positivity and uniqueness conditions. In particular, for LTRC and for random double truncation, identification from the indicator class is established in
Propositions~\ref{prop:S1} and~\ref{prop:S2} in Section~\ref{sec:examples}, respectively. For
current status data, the analogous converse result is given in Proposition~\ref{prop:S3} in
Supplementary Material Section~\ref{sec:CS_supp}.

\subsection{RKHS indexing and regularity (Assumption A1)}

Let $\mathcal S_{Z_1}$ be the support of $Z_1$ and let $\mathcal S_K\subseteq\mathcal S_{Z_1}$.
Let $K:\mathcal S_K\times\mathcal S_K\to\mathbb R$ be a positive definite kernel with RKHS
$\mathcal H_K$, inner product $\langle\cdot,\cdot\rangle_K$, and norm $\|\cdot\|_K$. Let $K_z(z'):=K(z,z')$ and $\mathcal B_K:=\{\varphi\in\mathcal H_K:\ \|\varphi\|_K\le 1\}$. We assume $\mathcal H_K\subset L_2(P)$, which by Proposition 2.5 in \cite{FerreiraMenegatto2013} is guaranteed under
\begin{equation}\label{eq:kernel-bounded}
E[K(Z_1,Z_1)]<\infty.
\end{equation}

\noindent We base inference on regular (pathwise differentiable) moments in the sense of
\cite{VdVaart1991}. Let $\hat g_\varphi$ denote an estimator of $g_\varphi(\cdot,\theta_0,G)$ under
$\mathcal H_0$ (typically $\hat g_\varphi(Z_1)=g_\varphi(Z_1,\hat\theta,\hat{G})$). The next assumption
formalizes the first-order influence-function (IF) expansion and is used throughout the paper.

\medskip
\noindent\textbf{Assumption A1 (Regularity / IF expansion).}
For each $\varphi\in\mathcal H_K$, under $\mathcal H_0$,
\[
E_n[\hat g_\varphi(Z_1)]
=
E_n[\Pi^\perp \varphi(Z_1)] + o_P(n^{-1/2}),
\]
where $\Pi^\perp:\mathcal H_K\to L_2(P)$ is linear and continuous and satisfies
$E_{\theta_0,G}[\Pi^\perp \varphi(Z_1)]=0$.

We illustrate the verification of Assumption A1 with random sampling with complete
observations, which we use as a running example throughout the article.

\begin{example}[Random sampling from $F$]\label{ex:running-rs}
Under random sampling, $X = Z = Z_1$ and $B \equiv 1$. For the MLE $\hat{\theta}$ of
$\theta_0$, define
\begin{equation}
\hat{g}_\varphi(z_1) = g_\varphi(z_1, \hat{\theta}) = \varphi(z_1) -
\int_{\mathcal{S}_{Z_1}} \varphi(z) \, dF_{\hat{\theta}}(z).
\label{example1}
\end{equation}
There is no $G$ in this example. Standard MLE theory yields Assumption A1 with
\begin{equation}
\Pi^\perp \varphi(z_1) = \varphi(z_1) - E_{\theta_0}[\varphi(Z_1)] +
B(\varphi, \theta_0) I_{\theta_0}^{-1} l_{\theta_0}(z_1),
\label{eq:pi-perp-rs}
\end{equation}
where $B(\varphi, \theta_0) = -E[\varphi(Z_1) l_{\theta_0}'(Z_1)]$,
$l_\theta(\cdot) = \partial \log f_\theta(\cdot)/\partial \theta$, $f_\theta$ is
the density of $F_\theta$, and $I_{\theta_0} = E[l_{\theta_0}(Z_1)l_{\theta_0}'(Z_1)]$
is the Fisher information matrix. $\square$
\end{example}

\subsection{Orthogonal score process, RKHS aggregation, and bootstrap}

Given an estimator $\hat\Pi^\perp$ of $\Pi^\perp$, define the orthogonal score process $\widehat{\Delta}_n^\perp(\varphi) := E_n[\hat{\Pi}^\perp \varphi(Z_1)]$, $\varphi\in\mathcal H_K$. To obtain an omnibus test, we aggregate over the RKHS unit ball: $\widehat{Q}_K
:=
\sup_{\varphi\in\mathcal B_K}\left(\widehat{\Delta}_n^\perp(\varphi)\right)^2$. Critical values are computed via a multiplier bootstrap:
\[
\widehat{Q}_K^*
:=
\sup_{\varphi\in\mathcal B_K}\left(\widehat{\Delta}_n^{\perp*}(\varphi)\right)^2,
\qquad
\widehat{\Delta}_n^{\perp*}(\varphi)
:=
\frac{1}{n}\sum_{i=1}^n \omega_i\,\hat{\Pi}^\perp\varphi(Z_i),
\]
where $\{\omega_i\}_{i=1}^n$ are iid\ multipliers with mean zero and unit variance, e.g.,
the two-point distributions in \cite{Wu1986,Mammen1993}. A key computational feature is that
nuisance parameters need not be re-estimated in the bootstrap resamples. Section~\ref{Computation}
derives a quadratic-form representation of $\widehat{Q}_K$ and $\widehat{Q}_K^*$.
Asymptotic validity under $\mathcal H_0$ and local alternatives is established in
Section~\ref{Sec:Asymptotics}.

%%%%%%%%%%%%%%%%%%%%%%%%%%%%%%%%%%%%%%%%%%%%%%%%%%%%%%%%%%%%%%%%%%%%%%%%%%%%%%
\section{General asymptotic results}\label{Sec:Asymptotics}

\subsection{Main result}\label{sec:asymp_main}

This section derives the generic asymptotic null behavior of the RKHS-aggregated,
orthogonal score statistic introduced in Section~\ref{sec:setup}. We maintain the
notation and standing conditions from Section~\ref{sec:setup}; in particular, $K$ is a
positive definite kernel on $\mathcal S_K\subseteq \mathcal S_{Z_1}$ with RKHS
$\mathcal H_K$ and unit ball $\mathcal B_K$, and $\mathcal H_K\subset L_2(P)$, e.g., under
\eqref{eq:kernel-bounded}. We also assume $\mathcal H_K$ is separable (e.g., if
$\mathcal S_K$ is separable and $K$ is continuous; see \cite[Lemma 4.33]{Christmann2008}),
so that its dual $\mathcal H_K^\ast$ is separable.

Let $\mathcal H_K^\ast$ denote the space of bounded linear functionals on $\mathcal H_K$,
equipped with the operator norm $\|T\|:=\sup_{\varphi\in\mathcal B_K}|T(\varphi)|$.
Define the empirical process $\mathbb{G}_n f := \sqrt{n}\{E_n[f]-E[f]\}$, $f\in L_2(P)$. Let $g_\varphi(\cdot,\theta_0,G)$ be the (full or conditional) score constructed in
Section~\ref{sec:full_cond_scores}, and let $\hat g_\varphi$ be an estimator under
$\mathcal H_0$ (typically $\hat g_\varphi(Z_1)=g_\varphi(Z_1,\hat\theta,\hat{G})$).
Assumption~A1 (stated in Section~\ref{sec:setup}) postulates an IF expansion
with an orthogonal map $\Pi^\perp:\mathcal H_K\to L_2(P)$. The statistic is formed from an
estimator $\hat\Pi^\perp$ of $\Pi^\perp$.

\medskip
\noindent \textbf{Assumption A2.}
There exists $\hat{\Pi}^\perp$ such that, uniformly in $\varphi \in \mathcal{B}_K$,
under $\mathcal{H}_0$,
\[
E_n[\hat{\Pi}^\perp \varphi(Z_1)] = E_n[\Pi^\perp \varphi(Z_1)] + o_P(n^{-1/2}).
\]

\medskip
\noindent \textbf{Assumption A3.}
$\varphi\in\mathcal H_K \rightarrow S_n(\varphi) := \mathbb{G}_n\{\Pi^\perp \varphi(Z_1)\}$ satisfies $S_n \Longrightarrow S_\infty$ in $\mathcal{H}_K^\ast$.

Define $\hat{\Delta}^\perp_n(\varphi) := E_n[\hat{\Pi}^\perp \varphi(Z_1)]$ and $\hat{Q}_K := \sup_{\varphi\in\mathcal B_K}\big(\hat{\Delta}^\perp_n(\varphi)\big)^2$.

\begin{theorem}\label{Thm:GeneralNull}
Under Assumptions A1--A3 and $\mathcal{H}_0$, it holds $n \hat{Q}_K \rightarrow_d \|S_\infty\|^2$.
\end{theorem}

Theorem~\ref{Thm:GeneralNull} is the generic device used in the examples of
Section~\ref{sec:examples}: Assumption~A1 provides the orthogonal IF map
$\Pi^\perp$, Assumption~A2 provides a uniformly valid plug-in estimator $\hat\Pi^\perp$, and
Assumption~A3 delivers the weak limit of the resulting RKHS-indexed empirical process.

%\noindent Assumption A1 guarantees the existence of the influence function $\Pi^{\perp }\varphi (Z_{1})$, with Assumption A2 providing an estimator and the form of the test statistic.

Assumption A1 is easier to verify when
$E_{n}[\hat{g}_{\varphi }(Z_{1})]$ depends only on empirical distributions of
observed data. The fact that $\Pi^{\perp }\varphi (Z_{1})$ is an influence function
explains why Assumption A2 may hold; see, for example, \cite[p.~396]{Bickel1993},
Theorem~6.1(i) in \cite[p.~557]{Huang1996}, Section~25.8 in \cite{VdVaart1998},
Assumption~H$_2$ in \cite{Bertail2006}, or Condition~M2 in \cite{Bickel2006}.
In general, $\hat{\Pi}^{\perp }$ depends on both $\hat{\theta}$ and $\hat{G}$ but,
in the relevant examples with truncation and censoring considered in
Section~\ref{sec:examples}, the nonparametric nuisance $\hat{G}$ does not appear
or it reduces to an empirical average with respect to the observable variables;
this facilitates things. In our leading examples
of Section~\ref{sec:examples} the score $g_{\varphi }(z,\theta ,G)$ will be free of
$G$, simplifying the computation of the influence function in Assumption A1.
Of course, Assumption A2 is empty in the simple null case.
Lemma~\ref{Lemma:B2} and Lemma~\ref{Lemma:B3} in Section~\ref{A2A3} provide sufficient
conditions to verify Assumption A2 and Assumption A3 in applications, respectively.

\subsubsection{Induced kernel}

Let $A$ denote a Gaussian Process with zero mean and covariance kernel $K$, in short
$A \sim GP(K)$, and let
$K_{g}^{\perp}(z_{1},z_{2}) = E_A[\Pi^{\perp}A(z_1)\Pi^{\perp}A(z_2)]$
be the induced covariance kernel of $\Pi^{\perp}A$, where $E_A$ denotes expectation
under $A \sim GP(K)$; see \cite{Berlinet2004, Matsumoto2023} and the specific
examples below.

To simplify exposition, while covering all examples of Section~\ref{sec:examples},
we consider the case where $g_\varphi(Z_1, \theta_0, G) = g_\varphi(Z_1, \theta_0)$
does not depend on $G$, and $\Pi^\perp$ has the representation
\begin{equation}
\Pi^\perp \varphi(z_1) = g_\varphi(z_1, \theta_0) -
\int_{\mathcal{S}_{Z_1}} g_\varphi(z, \theta_0)\, q_{\theta_0, G}(z, z_1)\, dP_{\theta_0, G}(z)
\label{eq:Pi-perp-rep}
\end{equation}
for some known function $q_{\theta_0,G}$ (up to $\theta_0$ and $G$). Henceforth,
all integrals are assumed well-defined; required integrability conditions for $K$
are collected in the next section.

Let $K_g(z_1, z_2) = E_A[TA(z_1)TA(z_2)]$ denote the induced kernel of $TA$, where
$T\varphi=g_\varphi(\cdot, \theta_0)$ is a mapping from $\mathcal{H}_K$ to $L_2(P)$.
We provide the expression for $K_g$ for several leading examples in
Section~\ref{sec:examples}, and focus now on the more complicated problem of obtaining
the induced kernel $K_g^\perp$ from $K_g$. Our next result follows readily from Fubini,
and hence its proof is omitted.

\begin{lemma}\label{Lemma:InducedKernel}
The kernel $K_{g}^{\perp}$ induced by $\Pi^\perp A$, with $\Pi^\perp$ given in
\eqref{eq:Pi-perp-rep}, is
\begin{align*}
K_g^\perp(z_1, z_2)
&= K_g(z_1, z_2)
- \int_{\mathcal{S}_{Z_1}} K_g(z_1, z)\, q_{\theta_0, G}(z, z_2)\, dP_{\theta_0, G}(z) \\
&\quad - \int_{\mathcal{S}_{Z_1}} K_g(z_2, z)\, q_{\theta_0, G}(z, z_1)\, dP_{\theta_0, G}(z) \\
&\quad + \int_{\mathcal{S}_{Z_1}} \int_{\mathcal{S}_{Z_1}} K_g(z, \tilde{z})\,
q_{\theta_0, G}(z, z_1)\, q_{\theta_0, G}(\tilde{z}, z_2)\,
dP_{\theta_0, G}(z)\, dP_{\theta_0, G}(\tilde{z}).
\end{align*}
\end{lemma}

\subsubsection{Sufficient conditions for Assumptions A2 and A3}
\label{A2A3}

Assumptions A2 and A3 require asymptotic results for empirical processes indexed by
elements in the RKHS $\mathcal{H}_K$. In Supplementary Material Appendix~\ref{ULLN}, we provide a new ULLN for
sample means indexed by both $\varphi \in \mathcal{B}_K$ and
$\theta \in \Theta_0 \subset \Theta$, where $\Theta_0$ is a compact set containing
$\theta_0$. This result (Lemma~\ref{lem:B1} in Supplementary Material Appendix~\ref{ULLN}) is useful to verify
our asymptotic results and it is of independent interest.

Consider a generalized version of \eqref{eq:pi-perp-rs}:
\begin{equation}
\Pi^\perp \varphi(z_1)=g_\varphi(z_1,\theta_0) + B(\varphi,\theta_0) I_{\theta_0}^{-1} l_{\theta_0}(z_1),
\label{eq:Pi}
\end{equation}
where
\[
B(\varphi, \theta) = E\left[\frac{\partial g_\varphi(Z_1, \theta)}{\partial \theta'}\right],
\qquad
l_\theta(z_1) = \frac{\partial}{\partial \theta} \log f_{\theta,G}(z_1),
\]
$f_{\theta,G}$ denotes the density of $Z_1$, possibly conditional on truncating or censoring
variables (see Section~\ref{sec:examples} for details), and
$I_{\theta_0} = E[l_{\theta_0}(Z_1) l_{\theta_0}'(Z_1)]$ is the Fisher information matrix,
assumed non-singular. We drop the dependence on $G$ of $l_{\theta}$ and $I_{\theta_0}$ for
simplicity of notation (but we allow these quantities to depend on $G$). An estimator of $\Pi^\perp \varphi(z_1)$ is
\begin{equation}
\hat{\Pi}^\perp \varphi(z_1) = g_\varphi(z_1, \hat{\theta}) + B_n(\varphi, \hat{\theta}) \hat{I}^{-1} \hat{l}(z_1),
\label{eq:Pihat}
\end{equation}
where $\hat{\theta}$ is a consistent estimator of $\theta_0$, and $B_n(\varphi, \hat{\theta})$
and $\hat{I}^{-1}\hat{l}(\cdot)$ are consistent estimators of $B(\varphi, \theta_0)$ and
$I_{\theta_0}^{-1}l_{\theta_0}(\cdot)$, respectively. We provide regularity conditions for Assumption A2 to hold. Define
$B_n^\partial(\varphi, \theta) = E_n[\partial g_\varphi(Z_1, \theta)/\partial \theta']$.
We allow but do not require that $B_n=B_n^\partial$. Let $int(\Theta_0)$ denote the interior
of $\Theta_0$.

\medskip
\noindent \textbf{Assumption B2.}
(i) For each $z \in \mathcal{S}_{Z_1}$ and $\varphi \in \mathcal{B}_K$,
$\theta \mapsto g_\varphi(z, \theta)$ is differentiable in the compact set $\Theta_0$
such that for any consistent estimator $\theta_n$ of $\theta_0\in int(\Theta_0)$,
$\|B_n^\partial(\cdot, \theta_n) - B(\cdot, \theta_0)\| = o_P(1)$.
Moreover, $\|B_n(\cdot, \hat{\theta}) - B(\cdot, \theta_0)\| = o_P(1)$;
(ii) $\hat{I}$ and $\hat{l}(\cdot)$ satisfy
\[
\left| \hat{I}^{-1} E_n[\hat{l}(Z_1)] - I_{\theta_0}^{-1} E_n[l_{\theta_0}(Z_1)] + (\hat{\theta} - \theta_0) \right| = o_P(n^{-1/2}),
\]
with $\left| E_n[l_{\theta_0}(Z_1)] \right| = O_P(n^{-1/2})$ and
$\left| \hat{\theta} - \theta_0 \right| = O_P(n^{-1/2})$.

Lemma~\ref{lem:B1} in Supplementary Material Appendix~\ref{ULLN} can be used to verify B2(i).
Assumption B2(ii) is standard.

\begin{lemma}\label{Lemma:B2}
Under Assumption~B2, Assumption A2 holds.
\end{lemma}

\paragraph{Continuation of Example~\ref{ex:running-rs}: checking Assumption A2.}
Standard MLE theory yields,
\[
\frac{\partial g_\varphi(z_1, \theta)}{\partial \theta'} = - \frac{\partial E_\theta[\varphi(Z_1)]}{\partial \theta'}
= - E_\theta[\varphi(Z_1) l_\theta'(Z_1)],
\]
which does not depend on $z_1$. Hence, $B_n^\partial(\varphi, \theta) = B(\varphi, \theta)$ and
sufficient conditions for $\|B_n^\partial(\cdot, \theta_n) - B(\cdot, \theta_0)\| = o_P(1)$ are that
$f_\theta$ is differentiable with
\[
\theta \mapsto \int_{\mathcal{S}_{Z_1}} K^{1/2}(z,z) \frac{\partial}{\partial \theta} f_\theta(z)\,dz
\quad\text{continuous in } \Theta_0;
\]
\begin{equation}
\int_{\mathcal{S}_{Z_1}} K^{1/2}(z,z) \sup_{\theta \in \Theta_0} \left| \frac{\partial}{\partial \theta} f_\theta(z) \right| dz < \infty.
\label{Khalfmoment}
\end{equation}
If $B_n(\varphi, \hat{\theta}) = -E_n[\varphi(Z_1) l_{\hat{\theta}}'(Z_1)]$,
Lemma~\ref{lem:B1} gives $\|B_n(\cdot, \hat{\theta}) - B(\cdot, \theta_0)\| = o_P(1)$ under
\begin{equation}
E\Big[K(Z_1,Z_1) \sup_{\theta \in \Theta_0} |l_\theta(Z_1)|^2\Big] < \infty.
\label{Kmoment}
\end{equation}
If $K$ is bounded, this reduces to a uniform second-order moment condition on the score $l_\theta$,
which is standard in MLE theory. $\square$

\medskip

Lemma~\ref{Lemma:B3} below shows that a mild sufficient condition for Assumption A3 is
\begin{equation}
E[K_g^\perp(Z_1, Z_1)] < \infty.
\label{eq:BI}
\end{equation}
This holds if $K_g^\perp$ is bounded, or under \eqref{eq:kernel-bounded} if $\Pi^\perp$ is a
\emph{norm-reducing} mapping in $L_2(P)$, i.e., $\|\Pi^\perp f\|_2 \le \|f\|_2$, $f\in L_2(P)$. Orthogonal projection operators in $L_2(P)$ are an example.

\medskip
\noindent \textbf{Assumption B3.}
$\Pi^\perp: \mathcal{H}_K \to L_2(P)$ is linear and continuous with
$E_{\theta_0,G}[\Pi^\perp \varphi(Z_1)] = 0$.

\begin{lemma}[Sufficiency for Assumption A3]\label{Lemma:B3}
(i) If Assumption~B3 holds, then \eqref{eq:BI} is sufficient for Assumption A3.
Furthermore, (ii) if $\Pi^\perp$ is norm-reducing under $\mathcal{H}_0$, then
\eqref{eq:kernel-bounded} is sufficient for \eqref{eq:BI} and $\Pi^\perp$ is continuous.
\end{lemma}

\noindent The continuity of $\Pi^\perp$ follows from
$\|\Pi^\perp f\|_2 \le \|f\|_2 \le C\|f\|_K$,  for each $f\in\mathcal{H}_K$,
where $C=\sqrt{E[K(Z_1, Z_1)]}$, see Proposition 2.5 in \cite{FerreiraMenegatto2013}.

\paragraph{Continuation of Example~\ref{ex:running-rs}: checking Assumption A3.}
Since $\Pi^\perp \varphi(x)$ is the orthogonal projection of $\varphi(z_1)$ onto the orthocomplement of
the span of $(1, l_{\theta_0}(z_1))$ in $L_2(P)$, $\Pi^\perp$ is norm-reducing. Hence, $E[K(Z_1, Z_1)]<\infty$ implies Assumption A3. $\square$

\subsection{Computation of the test}\label{Computation}

Let $\hat{\Pi}^\perp$ be given as in \eqref{eq:Pihat}, and assume the following.

\medskip
\noindent\textbf{Assumption C1.}
For all $\theta\in\Theta$ there exist $a_{\theta}:\mathcal{S}_{Z_1}\to\mathcal{H}_K$,
$b_\theta\in\mathcal{H}^p_K$, and $b_{n,\theta}\in\mathcal{H}^p_K$ such that, for all
$\varphi\in\mathcal{H}_K$ and $z\in\mathcal{S}_{Z_1}$,
\[
g_\varphi(z,\theta)=\langle \varphi,a_{\theta}(z)\rangle_K,\qquad
B(\varphi,\theta)=\langle \varphi,b_\theta\rangle_K,\qquad
B_n(\varphi,\theta)=\langle \varphi,b_{n,\theta}\rangle_K.
\]

\medskip % <-- space after C1 (as requested)

\noindent Define $\hat{b}=b_{n,\hat{\theta}}$, $\hat{h}(z)=a_{\hat{\theta}}(z)+\hat{b}\hat{I}^{-1}\hat{l}(z)$,
and for $z_1,z_2\in\mathcal{S}_{Z_1}$, $\hat{K}_{g}^{\perp}(z_1,z_2):=\langle \hat{h}(z_1),\hat{h}(z_2)\rangle_K$.

\begin{lemma}
\label{Lemma:QuadraticForm}
Under \eqref{eq:kernel-bounded} and Assumption~C1, \eqref{eq:QK_quadratic_form} holds.
\end{lemma}

\paragraph{Algorithm (generic computation of $n\hat Q_K$ and multiplier bootstrap).}
\begin{enumerate}[leftmargin=*,itemsep=0pt,topsep=2pt]
\item \textit{Estimate nuisances under $\mathcal H_0$:} compute $\hat\theta$ (and $\hat l,\hat I,\hat b$).
\item \textit{Build $\hat h$:} for each $i\le n$, evaluate $\hat h(Z_i)=a_{\hat\theta}(Z_i)+\hat b\,\hat I^{-1}\hat l(Z_i)\in\mathcal H_K$.
\item \textit{Form the Gram matrix:} compute $\hat K_g^\perp(Z_i,Z_j)=\langle \hat h(Z_i),\hat h(Z_j)\rangle_K$ for all $i,j$.
\item \textit{Compute the statistic:} $\hat Q_K = n^{-2}\sum_{i,j}\hat K_g^\perp(Z_i,Z_j)$.
\item \textit{Bootstrap:} for $b=1,\dots,B$, draw multipliers $\{\omega_i^{(b)}\}_{i=1}^n$ i.i.d.\ (e.g.\ Mammen or Rademacher) and set
$\hat Q_K^{*(b)} = n^{-2}\sum_{i,j}\omega_i^{(b)}\omega_j^{(b)}\hat K_g^\perp(Z_i,Z_j)$.
\item \textit{Decision:} reject if $\hat Q_K$ exceeds the $(1-\alpha)$ bootstrap quantile of $\{\hat Q_K^{*(b)}\}_{b=1}^B$.
\end{enumerate}

\paragraph{Continuation of Example~\ref{ex:running-rs}: checking Assumption C1.} In this example,
\[
\hat{\Pi}^{\perp}\varphi(z_1)=g_\varphi(z_1,\hat{\theta})+B_n(\varphi,\hat{\theta})\hat{I}^{-1}l_{\hat{\theta}}(z_1),
\]
where $g_\varphi(z_1,\hat{\theta})=\varphi(z_1)-E_{\hat{\theta}}[\varphi(Z_1)]$,
$B_n(\varphi,\hat{\theta})=-E_n[\varphi(Z_1)l'_{\hat{\theta}}(Z_1)]$, and
$\hat I=E_n[l_{\hat{\theta}}(Z_1)l'_{\hat{\theta}}(Z_1)]$. Assumption~C1 holds with $a_\theta(z)=K_z-\mu(F_\theta),$ $b_{\theta}(\cdot)=-E[K_{Z_1}(\cdot)l'_{\theta}(Z_1)],$ and $b_{n,\theta}(\cdot)=-E_n[K_{Z_1}(\cdot)l'_{\theta}(Z_1)],$
where $\mu(F_\theta):=\int K_u\,dF_\theta(u)\in\mathcal H_K$ (e.g.\ under \eqref{Khalfmoment}). A simple matrix representation for the test statistic follows from the Algorithm above and is provided in Section \ref{sec:Complete} of the Supplementary Material. $\square$
%%%%%%%%%%%%%%%%%%%%%%%%%%%%%%%%%%%%%%%%%%%%%%%%%%%%%
\subsection{Asymptotic power properties}
\label{Power}

\subsubsection{Global Power}

To derive the global power against fixed alternatives, we first show that
under the alternative hypothesis
\begin{equation}
\sup_{\varphi \in \mathcal{B}_{K}}\left\vert E_{n}[\hat{\Pi}^{\perp }\varphi
(Z_{1})-\Pi ^{\perp }\varphi (Z_{1})]\right\vert =o_{P}(1),  \label{P1}
\end{equation}
where $\hat{\Pi}^{\perp }\ $and $\Pi ^{\perp }$ are defined in (\ref{eq:Pihat})
and (\ref{eq:Pi}), respectively. In this discussion, $\theta _{0}$ is the
probabilistic limit of $\hat{\theta},$ which is assumed to exist under fixed
alternatives.

\medskip
\noindent \textbf{Assumption P1.} (i) Assumption B2(i) holds;  (ii) the estimators $\hat{I}^{-1}\hat{l}(\cdot
)$ satisfy
\begin{equation*}
\left\vert \hat{I}^{-1}E_{n}[\hat{l}(Z_{1})]-I_{\theta
_{0}}^{-1}E_{n}[l_{\theta _{0}}(Z_{1})]\right\vert =o_{P}(1);
\end{equation*}
$\left\vert E_{n}[l_{\theta _{0}}(Z_{1})]\right\vert =O_{P}(1)$ and $%
\left\vert \hat{\theta}-\theta _{0}\right\vert =o_{P}(1)$.\bigskip

Assumption P1(ii) is also standard in MLE theory, see \cite{Bickel1993}. The proof of the next
result follows the arguments of Lemma \ref{Lemma:B2}, and hence it is
omitted.
\begin{lemma}\label{lem:P1}
Under Assumption~P1, the uniform expansion in \eqref{P1} holds.
\end{lemma}

\begin{corollary}\label{cor:P1}
Under Assumption P1 and $E[K_{g}^{\perp
}(Z_{1},Z_{1})]<\infty $, it follows that
\begin{equation*}
\hat{Q}_{K}\rightarrow _{p}E[K_{g}^{\perp }(Z_{1},Z_{2})].
\end{equation*}
\end{corollary}
Consistency against all fixed alternatives follows if $E_{F,G}[K_{g}^{\perp
}(Z_{1},Z_{2})]>0$ (dependence on $F,G$ emphasized).
The main issues can be seen from our running example.

\paragraph{Continuation of Example~\ref{ex:running-rs}: consistency.} A useful interpretation for $E[K_{g}(Z_{1},Z_{2})]$ is as a
quadratic distance between $F$ and $F_{0}$:
\begin{eqnarray*}
d_{K}(F,F_{0}) &=&E[K_{g}(Z_{1},Z_{2})] \\
&=&\int_{\mathcal{S}_{Z_1}} \int_{\mathcal{S}_{Z_1}} K(z_{1},z_{2})d(F-F_{0})(z_{1})d(F-F_{0})(z_{2}),
\end{eqnarray*}
see \cite{Lindsay2008}. Therefore, $E[K_{g}(Z_{1},Z_{2})]>0$ under the
alternative hypothesis if the kernel $K$ is characteristic; see \cite{Sriperumbudur2011}. An alternative route to prove consistency considers a kernel having the integral representation $K(z_1,z_2)=\int_{\mathcal{T}}  k(z_{1},t)k(z_{2},t)d\Psi(t)$, for a class ${k(\cdot,t):t\in\mathcal{T}}$, an index set $\mathcal{T}$ and a measure $\Psi$, so that, by Fubini,
\begin{equation*}
d_{K}(F,F_{0})= \int_{\mathcal{T}}  (E_{F}[k(Z_{1},t)]-E_{F_{0}}[k(Z_{1},t)])^2d\Psi(t)=0,
\end{equation*}
iff $E_{F}[k(Z_{1},\cdot)]=E_{F_{0}}[k(Z_{1},\cdot)]$ $\Psi-a.s.$. To illustrate this consistency result, take $K=K_{ind}$, where $K_{ind}(z_1,z_2)=E_{F_0}\!\left[\1\{Z\le z_1\}\1\{Z\le z_2\}\right]$, and note that by Fubini,
\begin{equation*}
d_{K}(F,F_{0})=E_{F_{0}}[(F(Z_{1})-F_{0}(Z_{1}))^{2}],
\end{equation*}
so $d_{K}(F,F_{0})=0$ iff $F=F_{0}$ $F_{0}$-a.s. $\square$

We now show that under mild conditions, for the composite general case $E[K_{g}^{\perp }(Z_{1},Z_{2})]=E[K_{g}(Z_{1},Z_{2})]=:d_{K}(F,F_{0})$ provided $E\left[ l_{\theta _{0}}(X_{1})\right] =0$.

\begin{theorem}\label{thm:P1}
Under Assumptions P1 and C1, and $E[l_{\theta_0}(Z_1)] = 0$,
$E[K_{g}^{\perp }(Z_{1},Z_{2})]=d_{K}(F,F_{0})= \|E\left[a_{\theta_0}(Z_1)\right]\|^2_K$. Hence, the test is consistent, provided that $d_{K}(F,F_{0})>0$.
\end{theorem}

\noindent We compute the distance $d_{K}(F,F_{0})$ for some leading examples in Section \ref{sec:examples}.

\subsubsection{Local Power}

Let $H$ denote a cdf different from $F_{0}$, and consider the mixture model $%
F_{0,\tau ,H}=(1-\tau )F_{0}+\tau H$, $\tau \in \lbrack 0,1]$. Define the
local Pitman alternatives $\mathcal{H}_{0,n}:F=F_{0,\tau _{n},H},\text{ }\tau _{n}=1/\sqrt{n}$. The asymptotic local power properties of the proposed test will depend on
the slope
\begin{equation*}
D_{H}(\varphi ):=\left. \frac{\partial E_{F_{\tau },G}[\Pi ^{\perp }\varphi
(Z_{1})]}{\partial \tau }\right\vert _{\tau =0}.
\end{equation*}
For that reason we require smoothness of the drift function.

\medskip
\noindent \textbf{Assumption A4.} For each $\varphi \in \mathcal{H}_{K}$,
the function $E_{F_{\tau },G}[\Pi ^{\perp }\varphi
(Z_{1})] \ $is continuously differentiable at $\tau =0$, with
a derivative $D_{H}$ such that $D_{H}\in \mathcal{H}_{K}^{\ast }$.

\medskip

A sufficient condition for the differentiability in Assumption A4 is
\[
\lim\sup_{\tau\downarrow0}%
%TCIMACRO{\tint }%
%BeginExpansion
{\textstyle\int}
%EndExpansion
(\Pi ^{\perp }\varphi
(z_1))^{2}(z_1)dP_{F_{\tau },G}(z_1)<\infty;
\]
see Lemma~7.2 in \cite[p.~67]{IbragimovKhasminskii1981}. The condition then holds with $D_{H}(\varphi )=E[\Pi ^{\perp }\varphi
(Z_{1})s(Z_{1})]$, where $s(Z_{1})$ is the score of $P_{F_{\tau },G}$ at $\tau=0$. If $\Pi ^{\perp }$ is norm reducing, $s\in L_2(P)$ and (\ref{eq:kernel-bounded}) holds, then $D_{H}\in \mathcal{H}_{K}^{\ast }$ and Assumption A4 is satisfied.

\begin{theorem}\label{thm:P2} Let Assumptions A1-A4 hold under $\mathcal{H}_{0,n}.$
Then, under $\mathcal{H}_{0,n}$ it holds $n\hat{Q}_{K}\rightarrow _{d}\left\Vert R_{\infty }+D_{H}\right\Vert ^{2}$.
\end{theorem}
The verification of Assumptions A1-A4 under $\mathcal{H}_{0,n}$ and the
proof of Theorem \ref{thm:P2} follow the same arguments as those for Theorem \ref{Thm:GeneralNull} under $%
\mathcal{H}_{0}$, without additional complications, and hence these are omitted. The asymptotic local
power function will be non-trivial as long as $\left\Vert D_{H}\right\Vert
\neq 0$, i.e. $D_{H}\neq 0$ in $\mathcal{H}_{K}^{\ast }$.
\begin{remark}
There is no optional choice of $K$ from a global power point of view, see, e.g., \cite{BierensPloberger1997}, but there is an optimal choice for a given fixed local alternative. The reader is referred to Remark 3 in \cite{Escanciano2024} for further discussion in a different context. $\square$
\end{remark}
%This will very
%much depend on each specific application. In the next section, we give the
%slope function $D_{H}$ for some leading examples.

\section{Examples: LTRC and random double truncation}
\label{sec:examples}

This section specializes the general framework to two canonical incomplete-data
designs that are central in survival analysis and truncation problems:
(i) left-truncation with right-censoring (LTRC), and (ii) random double truncation.
In each example we (a) specify a score $g_\varphi(z,\theta,G)$ satisfying
$E[g_\varphi(Z_1,\theta_0,G)]=0$ under $\mathcal{H}_0$, (b) derive the corresponding
orthogonal influence function $\Pi^\perp\varphi$, and (c) provide convenient
sufficient conditions for Assumptions~A2--A3. In both examples below the score
is free of $G$, and we therefore write $g_\varphi(z,\theta)$. This also applies to the setting with current status data, see the Supplementary Material.

\subsection{Left truncation with right censoring (LTRC)}
\label{subsec:ltrc}

Let $U$ and $C$ denote a left-truncating and a right-censoring variable, respectively.
In the LTRC setting one observes $Z=(X\wedge C,\;U,\;\1\{X\leq C\})$ when $U\leq X\wedge C$; otherwise nothing is observed. Hence, $W=(U,C)$,
$T(x,u,c)=(x\wedge c,u,\1\{x\le c\})$, and $B(x,u,c)=\1\{u\leq x\wedge c\}$.
It is assumed that $W=(U,C)$ is independent of $X$. For identifiability, we also
assume $P(U\leq C)=1$, so $B(x,u,c)$ reduces to $\1\{u\leq x\}$.

Introduce $Y=X\wedge C$ and $\delta=\1\{X\leq C\}$, so $Z=(Y,U,\delta)$.
The data are iid copies $Z_i=(Y_i,U_i,\delta_i)$, $1\le i\le n$, with the
conditional distribution of $Z$ given $U\le X$. Let $z=(y,u,d)$ and define
\[
g_\varphi(z,\theta)=\varphi(y)d-\psi^\varphi_\theta(z),
\qquad
\psi^\varphi_\theta(z)=\psi^\varphi_\theta(y,u,d)
=\int \varphi(\tilde{x})\frac{\1\{u\le \tilde{x}\le y\}}{1-F_\theta(\tilde{x}-)}\,dF_\theta(\tilde{x}).
\]
Then $E[g_\varphi(Z_1,\theta_0)]=0$ under the null. Conversely, if this equality holds
along a rich class of $\varphi$-functions, then $F=F_{\theta_0}$, provided that $M(x)\equiv P(U\le x\le C)$ is positive on the support of $X$, denoted $\mathcal{S}_X$. In absence of such a condition,
a portion of $F$ cannot be identified from the available data; hence this positivity
assumption is minimal.

For ease of presentation, we formalize the characterization for the class of indicator functions
$\mathcal{C}_{\mathrm{ind}}=\{\varphi_x(\cdot):x\in\mathcal{S}_X\}$, where
$\varphi_x(y)=\1\{y\le x\}$. Proofs for all results in this section are given in
Supplementary Material Appendix~\ref{ExamplesProofs}.

\begin{proposition}\label{prop:S1}
Under $F = F_{\theta_0}$, it holds that $E[g_{\varphi}(Z_1, \theta_0)] = 0$ for each
$\varphi \in \mathcal{C}_{\mathrm{ind}}$. Conversely, if $E[g_{\varphi}(Z_1, \theta_0)] = 0$
for each $\varphi \in \mathcal{C}_{\mathrm{ind}}$, then $F = F_{\theta_0}$, provided that
$M(x) > 0$ for each $x \in \mathcal{S}_X$.
\end{proposition}

The function $\psi^\varphi_\theta$ can be obtained as a conditional score by using the general rule (\ref{eq:gphi_cond}) as follows. Let $\psi(z)=\varphi(y)d$ and note that, for the LTRC model, $W_1^o=U_1$. Then, it can be seen that

$$E[\psi(Z_1)|U_1]=\int \varphi(x)\dfrac{M^*(x|U_1)}{1-F( x-)}dF( x)$$

\noindent where $M^*(x|u)=E[I(U_1\leq x\leq Y_1)|U_1=u]$. Since $M^*$ involves the censoring distribution and since the censoring variable is not observed for all the individuals, the value $M^*(x|U_1)$ is replaced by the observable indicator $I(U_1\leq x\leq Y_1)$ Finally, $F$ is replaced by the null model $F_\theta$ to get $\psi^\varphi_\theta(z)$.

From the definition of $g_\varphi(z,\theta)$ it is easy to see that
$B(\varphi,\theta)=-E[\varphi(Y_1)h_\theta(Y_1)\delta_1]$, where
\[
h_{\theta }(y)=\frac{\partial }{\partial \theta ^{\prime }}\log \left(
\frac{f_{\theta }(y)}{1-F_{\theta }(y-)}\right).
\]
With left truncation and right censoring, the score attached to the conditional
loglikelihood of $(Y_1,\delta_1)$ given $U_1=u$ is
\[
l_{\theta }(z)=l_{\theta }(y,u,d)=\frac{\partial }{\partial \theta }\log
\left( \frac{f_{\theta }(y)^{d}(1-F_{\theta }(y))^{1-d}}{1-F_{\theta }(u-)}\right).
\]
Let $\hat{\theta}$ be the solution of $E_n[l_{\hat{\theta}}(Z_1)]=0$. The influence function of the
process $E_n[g_\varphi(Z_1,\hat{\theta})]$ is
\begin{equation}
\Pi^\perp\varphi(z)
=
g_{\varphi}(z,\theta_0)
-
E[\varphi(Y_1)h_{\theta_0}(Y_1)\delta_1]\,
I_{\theta_0}^{-1}l_{\theta_0}(z).
\label{eq:LTRC_pi}
\end{equation}

We provide sufficient conditions for Assumptions~A2 and A3 in this example.

\medskip
\noindent \textbf{Assumption E1.}
For each $z_1=(y_1,u_1,d_1)\in\mathcal{S}_{Z_1}$ and $\theta\in\Theta_0$,
$\theta\mapsto h_\theta(y_1)$ is continuous in $\Theta_0$ a.s. and
\[
E\!\left[ K(Y_1,Y_1)\sup_{\theta\in\Theta_0}\vert h_\theta(Y_1)\vert^2\delta_1\right]<\infty.
\]
\noindent See \cite{EmuraMichimae2022} and references therein for details on MLE under LTRC to satisfy
Assumption~B2(ii).

\begin{lemma}\label{lem:E1}
Under Assumption E1 and Assumption B2(ii), Assumption A2 holds.
\end{lemma}

The norm-reducing property of \eqref{eq:LTRC_pi} is established next.

\begin{lemma}\label{lem:E2}
The function $\Pi^\perp$ in \eqref{eq:LTRC_pi} is norm reducing, so
$E[K_g^\perp(Z_1,Z_1)]<\infty$ holds under $E[K(Y_1,Y_1)\delta_1]<\infty$.
\end{lemma}

For global power, one needs to check whether $d_K(F,F_0)>0$ holds for $F\neq F_0$.
By straightforward calculations,
\[
d_{K}(F,F_{0})
=
\alpha^{-2}
\int \!\!\int K(x_{1},x_{2})R(x_{1})R(x_{2})
\,d(\Lambda_{F}(x_{1})-\Lambda_{F_{0}}(x_{1}))
\,d(\Lambda_{F}(x_{2})-\Lambda_{F_{0}}(x_{2})),
\]
where $\alpha =P(U\leq X)$, $R(x)=M(x)(1-F(x-))$, and $\Lambda_F$ is the cumulative hazard
pertaining to $F$. If $K$ is characteristic, see \cite{Sriperumbudur2011}, then $d_K(F,F_0)>0$
for $F\neq F_0$ provided that $M(x)>0$ for each $x\in\mathcal{S}_X$.\\

\medskip

\textbf{Comparison to existing tests for LTRC data.} The conditions for convergence of the proposed maximum mean discrepancy test are essentially minimal, and only imposed to guarantee the asymptotic properties of the maximum-likelihood estimator $\hat \theta$; see Assumption E1 above. This is in contrast to the existing NPMLE-based tests for censored data, or for LTRC data, which suffer from important limitations. Specifically, existing omnibus methods for censored data control the limit behaviour of the NPMLE by restricting weak convergence to a compact interval, bounded away from the upper limit of the observable distribution. As a consequence, the could fail to detect violations of the postulated model at the far right tail \citep{Nikabadze1997}. \cite{Sun1997} investigated the weak convergence of the Kaplan-Meier process on the whole support of $Z_1$. For this, the condition $\int (1-G)^{-1}dF<\infty$ was needed; see also \cite{Gill1983}. This and other integrability conditions involving the censoring distribution $G$ have been used in the literature to guarantee the correct behaviour of NPMLE-based tests; however, they may fail in practical settings. For instance, when $1-G=(1-F)^\beta$ for some $\beta\geq 1$ the aforementioned integral diverges. Existing NPMLE-based omnibus tests for left-truncated data have similar limitations; see for instance the strong integrability conditions on the truncation distribution in \cite{Sun1997} or \cite{StuteWang2008}. Our RKHS-based framework is free of such limitations.

\subsection{Random double truncation (doubly truncated data)}
\label{subsec:double_trunc}

With doubly truncated data one observes $Z=(X,U,V)$ when $U\le X\le V$; otherwise nothing is observed.
Here, $U$ and $V$ are the left- and right-truncating variables. Hence, $W=(U,V)$,
$T(x,u,v)=(x,u,v)$ (the identity), and $B(x,u,v)=\1\{u\le x\le v\}$.
It is assumed that $W=(U,V)$ is independent of $X$. For identifiability, also assume $P(U\le V)=1$.

The data are iid copies $Z_i=(X_i,U_i,V_i)$, $1\le i\le n$, with the conditional distribution of $Z$
given $U\le X\le V$. Let $z=(x,u,v)$ and define
\[
g_\varphi(z,\theta)=\varphi(x)-\psi^\varphi_\theta(z),
\qquad
\psi^\varphi_\theta(z)=\psi^\varphi_\theta(x,u,v)
=\int \varphi(\tilde{x})\frac{\1\{u\le \tilde{x}\le v\}}{F_\theta(v)-F_\theta(u-)}\,dF_\theta(\tilde{x}).
\]
Actually, $\psi^\varphi_\theta(z)$ is directly obtained from $\psi(z)=\varphi(x)$ by noting $W_1^o=(U_1,V_1)$ and using the general rule (\ref{eq:gphi_cond}). Then $E[g_\varphi(Z_1,\theta_0)\mid U_1,V_1]=0$, and hence $E[g_\varphi(Z_1,\theta_0)]=0$.
Conversely, when the latter equality holds along the class of indicators, then $F=F_{\theta_0}$.
This result requires $M(x)\equiv P(U\le x\le V)>0$ on the support of $X$; as in the LTRC case,
this condition avoids identifiability issues and is minimal. See \cite{deUnaAlvarezVanKeilegom2021}
for further discussion. Additionally, we require
\begin{equation}
\int_{-\infty}^x M\,dF=\int_{-\infty}^x M_0\,dF_0
\quad\text{for all } x\in\mathcal{S}_X
\quad\Longrightarrow\quad
F=F_0,
\label{key:DT}
\end{equation}
where
\[
M_0(x)=E\left[\frac{F(V)-F(U-)}{F_0(V)-F_0(U-)}\1\{U\le x\le V\}\right].
\]
Condition \eqref{key:DT} corresponds to uniqueness of the solution to $H(F_0,F^*,G^*)\equiv 0$,
where $H(F_0,F^*,G^*)$ is the operator in \cite{deUnaAlvarezVanKeilegom2021}, proof to Theorem~2.1,
and $F^*$ and $G^*$ are the observable versions of the target and truncation cdfs. The condition
rules out the possibility that two different target cdfs linked to the same observable couple $(F^*,G^*)$
may exist. \cite{deUnaAlvarezVanKeilegom2021} provide primitive conditions that ensure uniqueness
and thus \eqref{key:DT}. We formalize identification as follows.

\begin{proposition}\label{prop:S2}
Under $F = F_{\theta_0}$, it holds $E[g_{\varphi}(Z_1, \theta_0)] = 0$ for each
$\varphi \in \mathcal{C}_{\mathrm{ind}}$; conversely, if $E[g_{\varphi}(Z_1, \theta_0)] = 0$
for each $\varphi \in \mathcal{C}_{\mathrm{ind}}$, then $F = F_{\theta_0}$, provided that
$M(x) > 0$ for each $x \in \mathcal{S}_X$ and that \eqref{key:DT} holds.
\end{proposition}

The score $g_\varphi(z,\theta)$ above is different from that considered in the LTRC setting, even in
common scenarios such as one-sided (left) truncation or random sampling. None of these two possible
scores uniformly dominates the other; see Section~\ref{sub:sim_random} in the Supplementary Material for empirical results.

From the definition of $g_\varphi(z,\theta)$ it follows that
$B(\varphi,\theta)=-E[\varphi(X_1)l_\theta'(Z_1)]$, where the score attached to the conditional
loglikelihood of $X_1$ given $U_1=u$ and $V_1=v$ is
\[
l_\theta(z)=l_\theta(x,u,v)=\frac{\partial}{\partial\theta}\log\left(\frac{f_\theta(x)}{F_\theta(v)-F_\theta(u-)}\right).
\]
Let $\hat{\theta}$ be the solution of $E_n[l_{\hat{\theta}}(Z_1)]=0$. The influence function of the process
$E_n[g_\varphi(Z_1,\hat{\theta})]$ is
\begin{equation}
\Pi^\perp\varphi(z)
=
g_{\varphi}(z,\theta_0)
-
E[\varphi(X_1)l_{\theta_0}'(Z_1)]\,
I_{\theta_0}^{-1}l_{\theta_0}(z).
\label{eq:DT_pi}
\end{equation}

\noindent We consider the following assumptions; see Theorem~2.2.7 in \cite{deUnaAlvarez2021Book}
for detailed conditions in MLE under double truncation to ensure Assumption~B2(ii).

\medskip
\noindent \textbf{Assumption E2.}
For each $z_1=(x_1,u_1,v_1)\in\mathcal{S}_{Z_1}$ and $\theta\in\Theta_0$,
$\theta\mapsto l_\theta(z_1)$ is continuous in $\Theta_0$ and
\[
E\!\left[ K(X_1,X_1)\sup_{\theta\in\Theta_0}\vert l_\theta(Z_1)\vert^2\right]<\infty.
\]

\begin{lemma}\label{lem:E3}
Under Assumption E2 and Assumption B2(ii), Assumption A2 holds.
\end{lemma}

\noindent The norm-reducing property of \eqref{eq:DT_pi} is established next.

\begin{lemma}\label{lem:E4}
The function $\Pi^\perp$ in \eqref{eq:DT_pi} is norm reducing, so
$E[K_g^\perp(Z_1,Z_1)] < \infty$ holds under $E[K(X_1,X_1)] < \infty$.
\end{lemma}

\noindent For doubly truncated data, $d_K(F,F_0)$ equals
\[
\alpha^{-2}
\int\!\!\int K(x_{1},x_{2})
\big(M(x_1)dF(x_{1})-M_0(x_1)dF_0(x_{1})\big)
\big(M(x_2)dF(x_{2})-M_0(x_2)dF_0(x_{2})\big),
\]
where $\alpha =P(U\leq X\leq V)$ and $M_0$ is the function in \eqref{key:DT}. A global consistency
result requires that the double integral above is strictly positive for $F\neq F_0$.
This holds for characteristic kernels provided that \eqref{key:DT} holds.

%%%%%%%%%%%%%%%%%%%%%%%%%%%%%%%%%%%%%%%%%%%%%%%%%%%%%%%%%
%%%%%%%%%%%%%%%%%%%%%%%%%%%%%%%%%%%%%%%%%%%%%%%%%%%%%%%%%
\section{Numerical studies and illustration}\label{sec:num}

This section reports Monte Carlo evidence on the finite-sample performance of the
proposed test in the random double truncation design of Section~\ref{subsec:double_trunc},
and provides a real-data illustration for the same setting. Results for the running
example of random sampling from $F$ (complete data) are reported in the
Supplementary Material (see Appendix~S.\ref{SM:sim_random}).

\subsection{Simulation study: random double truncation}\label{sub:sim_double_trunc}

We consider the random double truncation design in Section~\ref{subsec:double_trunc}.
To generate double truncation we use an interval sampling scheme where $(U,V)$ lies on a
line. Specifically, for a constant $\nu$, the left-truncating variable is simulated as
$U\sim Exp(1)-\nu$, independently of $X$, and we set $V=U+3+\nu$. Hence, the sampling
window has width $3+\nu$, and truncation becomes more severe as $\nu$ decreases. We
consider $\nu=1$ (weak truncation, WT) and $\nu=0.5$ (strong truncation, ST). Truncated
observations are discarded and resampled until the desired sample size $n$ is reached. We use 1000 Monte Carlo trials under
$\mathcal{H}_0$ and 500 trials under $\mathcal{H}_1$, and sample sizes
$n\in\{50,100,200\}$.

We use the statistic $n\hat Q_K$ based on the orthogonal influence function
\eqref{eq:DT_pi}. The test functions range over the indicator class, and critical values
are obtained via the multiplier bootstrap with $B=499$ resamples. In the absence of
truncation this construction reduces to the random-sampling implementation for $n\hat Q_K^{(2)}$ reported in
the Supplementary Material, Section S2.

For implementation, the induced kernel $K_g^\perp$ must be estimated. This is partly
done analytically as in the complete-data case, but under composite nulls with double
truncation some expectations are difficult to evaluate in closed form and are therefore
replaced by sample averages.

Table~\ref{tab:DT} reports rejection proportions at level $\alpha=0.05$, together with
the proportion of truncated observations. The level is well respected in all
configurations. Under alternatives, power increases with sample size and with the
distance to the null. As expected, power decreases as truncation becomes stronger,
reflecting the information loss induced by the selection mechanism.

\begin{table}[tbp]
\begin{center}
{\small \addtolength{\tabcolsep}{-2pt}
\begin{tabular}{cccccccccc}
& WT &  &  &  &  &ST & &  &  \\
$\theta$ &pt & $n=50$ & $n=100$ & $n=200$ &  & pt& $n=50$ & $n=100$ & $n=200$ \\
\hline
& & & &  &  &  &  &  &  \\
$0.5$ & .266&.730 & .972 & 1.00 &  &.434 & .690 & .958 & 1.00 \\
$0.8$ & .226& .104 & .186 & .418 & &.364 & .104 & .176 & .360 \\
$1$ & .209&.052 & .053 & .066 &  &.328& .061 & .071 & .054 \\
$1.2$ & .197&.108 & .184 & .286 & &.302 & .114 & .164 & .238 \\
$1.5$ & .189&.330 & .526 & .864 & &.271 & .276 & .480 & .786 \\
&  &  &  & & & &  &  &  \\ \hline
\end{tabular}
}
\end{center}
\caption{Proportion of truncated data (pt) and rejection proportions of the test statistic $n\hat Q_K$ under weak
(WT) and strong (ST) random double truncation along the Monte Carlo trials at level
$\protect\alpha=0.05$. The null hypothesis corresponds to $\protect\theta=1$.}
\label{tab:DT}
\end{table}

\subsection{Empirical illustration: quasar data}\label{sub:quasar}

\cite{EfronPetrosian1999} investigated the distribution of the adjusted log luminosity
$X$ (after correction for luminosity evolution) using data on $n=210$ stellar objects.
These data are doubly truncated because only luminosities within an observation-specific
bounded interval can be detected. The data are available in the \texttt{DTDA} package in
\texttt{R} (object \texttt{Quasars}). \cite{EfronPetrosian1999} reported a nonlinear shape
of the estimated log survival function $\log(1-F_n(x))$, where $F_n$ denotes the NPMLE of
the cdf, suggesting lack of fit to the exponential model $\log(1-F(x))=\theta x$.

We revisit these data and formally test the composite null hypothesis that $X$ is
exponentially distributed. We apply the goodness-of-fit test adapted to random double
truncation as discussed in Section~\ref{subsec:double_trunc}, using the indicator class
$\varphi_x(\tilde x)=\1\{\tilde x\le x\}$ and $B=499$ multiplier bootstrap resamples.

In Figure~\ref{fig:Quasars_F}, left, we display the NPMLE $F_n$, together with the ordinary
empirical cdf of the observed $X_i$ for comparison. Following \cite{EfronPetrosian1999}, the
data are shifted by subtracting the minimum observed log luminosity so that they lie on the
positive half-line. The ordinary empirical cdf substantially overstates the distribution due
to the selection bias induced by double truncation, which makes low luminosities unlikely to
be observed. The MLE under the exponential model is $\hat\theta=1.7762$, and the fitted
exponential cdf is also shown.

Figure~\ref{fig:Quasars_F}, right, depicts the density of the bootstrapped statistics
$n\hat Q_K^\ast$, and the vertical dashed line indicates the value on the original data
($n\hat Q_K=0.0225$). The bootstrap p-value is $p=0.0661$. Thus, the exponential model is
accepted at the 5\% level but rejected at the 10\% level. This conclusion is consistent with
the curvature in the estimated log-survival function emphasized by \cite{EfronPetrosian1999},
and illustrates that the proposed orthogonal-score MMD statistic can be informative in doubly
truncated designs where nonparametric cdf estimation is intrinsically noisy.

\begin{figure}[htb]
\begin{center}
\includegraphics[width=140mm]{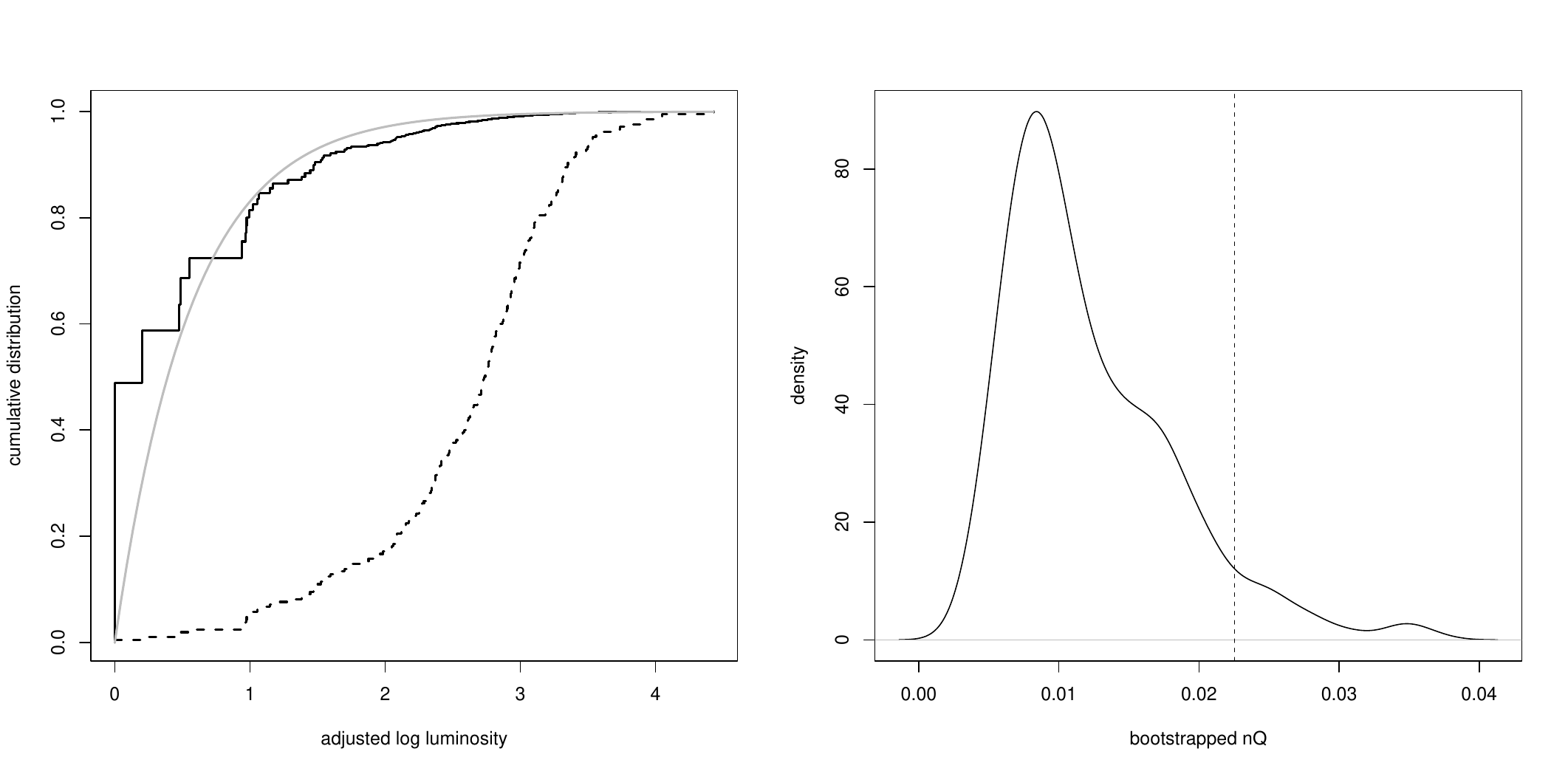}
\end{center}
\caption{Left: NPMLE of the quasar luminosity (solid black), ordinary empirical cdf (dashed), and fitted exponential model (gray). Right: density of $B=499$ multiplier-bootstrap values of $n\hat Q_K^\ast$; dashed line is $n\hat Q_K$ on the data ($p=0.0661$).}
\label{fig:Quasars_F}
\end{figure}

%%%%%%%%%%%%%%%%%%%%%%%%%%%%%%%%%%%%%%%%%%%%%%%%%%%%%%%%%
\section{Discussion}\label{sec:disc}

We proposed a general score-based framework for omnibus goodness-of-fit testing with censoring and truncation, combining (i) regular, IF-based orthogonalization and (ii) RKHS aggregation via an induced-kernel quadratic form. The resulting statistic is a computationally simple $V$-statistic, and the multiplier bootstrap keeps nuisance estimates fixed, which is particularly appealing in incomplete-data designs where re-estimating nuisance components can be unstable or expensive.

The LTRC and random double truncation examples show how the framework adapts by choosing unbiased scores and deriving the corresponding orthogonal influence functions. For characteristic kernels, the induced distance $d_K(F,F_0)$ provides a transparent route to global consistency under identifiability. The simulation evidence confirms accurate size and power patterns that reflect information loss as truncation intensifies, and the quasar illustration demonstrates practical relevance in doubly truncated data.

Several extensions are natural. First, developing principled data-driven kernel choices for improved sensitivity against application-specific local alternatives would be useful. Second, extending the framework to covariate-conditional GOF problems and dependent data (e.g., time series) would broaden applicability while retaining the same IF logic.
%%%%%%%%%%%%%%%%%%%%%%%%%%%%%%%%%%%%%%%%%%%%%%%%%%%%%%%%%
\section*{Acknowledgements}
Juan Carlos Escanciano gratefully acknowledges financial support from grants PID2021-127794NB-I00, CEX2021-001181-M, MCIN/AEI/10.13039/501100011033, Comunidad de Madrid (CAM) Grants EPUC3M11 (V PRICIT) and Mad-Econ-Pol-CM H2019/HUM-5891. Jacobo de U\~na-\'Alvarez is supported by Grant PID2023-148811NB-I00.
%%%%%%%%%%%%%%%%%%%%%%%%%%%%%%%%%%%%%%%%%%%%%%%%%%%%%%%%%%%
\bibliographystyle{abbrvnat}
\bibliography{references}

\ifIncludeSM
  \startSupplementaryMaterial
  \SMfrontmatter

%%%%%%%%%%%%%%%%%%%%%%%%%%%%%%%%%%%%%%%%%%%%%%%%%%%%%%%%%
%%%%%%%%%%%%%%%%%%%%%
\section{Additional results for complete data}
\label{sec:Complete}
In this section, we collect some additional results on the complete data example, including computation of the test statistic and comparison with some existing results.

\subsection{Computation of the test statistic}
Define
\begin{equation}
q_\theta(z, z_1) = l_\theta'(z) I_\theta^{-1} l_\theta(z_1), \quad
\Pi_\theta \varphi(z_1) = \int_{\mathcal{S}_{Z_1}} \varphi(z) q_\theta(z, z_1) dF_\theta(z).
\label{eq:q-operator}
\end{equation}
Then $\Pi^\perp \varphi(z_1) = g_\varphi(z_1, \theta_0) -
\Pi_{\theta_0} g_\varphi(z_1, \theta_0)$ since
$\int_{\mathcal{S}_{Z_1}} q_{\theta_0}(z, z_1) \, dF_{\theta_0}(z) = 0$ by the
zero-mean property of the score.

To compute the induced kernel in this example, note that $g_\varphi(z_1, \theta_0) = \varphi(z_1) - E_{\theta_0}[\varphi(X)]$, so that $K_g=K_{\theta_0}$, where
\begin{align*}
K_{\theta}(z_1, z_2)
&= K(z_1, z_2)
- \int_{\mathcal{S}_{Z_1}} K(z_1, z)\, dF_{\theta}(z)
- \int_{\mathcal{S}_{Z_1}} K(z_2, z)\, dF_{\theta}(z) \\
&\quad + \int_{\mathcal{S}_{Z_{1}}} \int_{\mathcal{S}_{Z_1}} K(z, \tilde{z})\, dF_{\theta}(z)\, dF_{\theta}(\tilde{z}).
\end{align*}
Using Lemma~\ref{Lemma:InducedKernel} and $q_{\theta_0, G} \equiv q_{\theta_0}$
from \eqref{eq:q-operator}, one can obtain $K_g^\perp$. The kernel $K_g^\perp$
coincides with the score-centered kernel of \citet{Lindsay2008}. 

From the induced kernel, the computation of the test statistic follows some simple matrix algebra. The $n \times n$ matrix $\mathbb{\hat{K}}_g^{\perp}$ with $ij$-th element $\hat{K}_g^{\perp}(Z_i,Z_j)$ is
\[
\mathbb{\hat{K}}_g^{\perp} = \hat{\Pi}_{l_{\hat{\theta}}}^{\perp} \mathbb{\hat{K}}_g \hat{\Pi}_{l_{\hat{\theta}}}^{\perp},
\]
where $\hat{\Pi}_{l_{\hat{\theta}}}^{\perp} = I_n - \mathbf{l}_{\hat{\theta}} (\mathbf{l}_{\hat{\theta}}' \mathbf{l}_{\hat{\theta}})^{-1} \mathbf{l}_{\hat{\theta}}$, $I_n$ is the $n\times n$ identity matrix, $\mathbf{l}_{\hat{\theta}}$ is the $n \times p$ matrix of scores with $i$-th row $l_{\hat{\theta}}'(Z_i)$, and $\mathbb{\hat{K}}_g$ is the $n\times n$ matrix with $ij$-th element $\hat{K}_g(z_i,z_j)=K_{\hat{\theta}}(z_i,z_j)$. To illustrate, with the kernel $K=K_{ind}$, we have the closed form
\[
\hat{K}_g(z_1,z_2)= 1 - F_{\hat{\theta}}(z_1 \vee z_2)
-\tfrac12\!\left(1 - F_{\hat{\theta}}(z_1)^2\right)
-\tfrac12\!\left(1 - F_{\hat{\theta}}(z_2)^2\right)+\tfrac13, 
\]
where $a\vee b=\max\{a,b\}$. 

\begin{remark}[Alternative test statistic]\label{remark:alternative}
For random sampling, an alternative statistic can be derived from the score function
\[
g_{\varphi}(z_1,\theta)
=
\varphi(z_1) - \int_{\mathcal{S}_{Z_1}} \varphi(z)\frac{\1\{z_1 \ge z\}}{1 - F_\theta(z-)}\, dF_\theta(z),
\]
as noted by \cite{Akritas1988} in the context of Pearson-type tests. For $K=K_{ind}$, the induced kernel is
\[
K_g(z_1,z_2)
=
E\Big[(\1\{z_1 \le Z_1\} - \Lambda_{F_0}(Z_1 \wedge z_1))
(\1\{z_2 \le Z_1\} - \Lambda_{F_0}(Z_1 \wedge z_2)) \Big],
\]
where $x \wedge y = \min(x,y)$ and $\Lambda_F$ denotes the cumulative hazard function attached to $F$. $\square$
\end{remark}

\subsection{Comparison with existing tests}
\subsubsection{Comparison with classical tests: simple hypothesis}

In the simple null hypothesis case, With $\theta_0$ known, choosing $K=K_{ind}$ yields the classical
distribution-function GOF family. By Lemma~\ref{Lemma:QuadraticForm} (with $\hat K_g^\perp=K_g$) and Fubini,
$n\hat Q_K$ coincides with the Cram\'er--von Mises statistic,
\[
n\hat Q_K=\int (\Delta_n(\varphi_z))^2\,dF_0(z),\qquad
\Delta_n(\varphi_z)=\sqrt n\,(\mathbb P_n\varphi_z-P_0\varphi_z),
\]
where $\varphi_z(Z)=\1\{Z\le z\}$ and $F_0=F_{\theta_0}$.
 
\subsubsection{Comparison with existing tests: composite hypothesis}

Our test for the complete data case improves upon existing tests in \cite{Lindsay2008,Lindsay2014} and \cite{key2021composite}. These papers propose as test statistic
\begin{equation}
\label{eq:comparison}
\tilde{Q}_{K}
=
\frac{1}{n^{2}}\sum_{i=1}^{n}\sum_{j=1}^{n}\hat{K}_{g}(Z_{i},Z_{j}).
\end{equation}

That is, existing procedures do not use orthogonal processes. \cite{Lindsay2008,Lindsay2014} allows for the MLE estimator, but uses a conservative Satterthwaite approximation to compute critical values. \cite{key2021composite} use a consistent parametric bootstrap approximation, but requires the estimator to be a minimum distance estimator (minimizing the MDD distance $\tilde{Q}_{K}$ when viewed as a function of $\theta$). This estimator is generally less efficient than MLE. Moreover, their test requires to re-estimate in each bootstrap sample, which may be computationally expensive. Our orthogonal process-based test solves these limitations: it is not conservative, allows for general estimators, including MLE, and it does not require to re-estimate in each bootstrap sample. 
%%%%%%%%%%%%%%%%%%%%%%%%%%%%%%%%%%%%%%%%%%%%%%%%%%%%%%%%%
\section{Current status data}
\label{sec:CS_supp}

\subsection{Observation scheme and null hypothesis}

Let $X$ denote the event time of interest with cdf $F$. In the current status design,
one observes an inspection time $C$ and the indicator
\[
\delta=\1\{X\le C\}.
\]
Assume $C$ is independent of $X$ with cdf $G$, and that the support of $G$ contains the
relevant region of the support of $X$ (positivity/overlap). The observed data are
$Z=(\delta,C)$, and we test the composite null
\[
\mathcal{H}_0:\quad F\in\{F_\theta:\theta\in\Theta\subset\R^p\}.
\]

\subsection{Score construction and identification}

For a test function $\varphi$ defined on the support of $C$, for $z=(d,c)$ define the score
\begin{equation}
g_\varphi(z,\theta)=g_\varphi(d,c;\theta)
=\varphi(c)\big(d-F_\theta(c)\big),
\label{eq:CS_score}
\end{equation}
so that $E[g_\varphi(Z_1,\theta_0)]=0$ under $\mathcal{H}_0$. Note that the score is free of $G$. Note also that the score (\ref{eq:CS_score}) follows from (\ref{eq:gphi_cond}) by setting $\psi(z)=\varphi(c)d$ and $W_1^o=C_1$.

Let $\mathcal{C}_{\mathrm{ind}}=\{\varphi_x(\cdot):x\in\mathcal{S}_X\}$ where
$\varphi_x(c)=\1\{c\le x\}$. The next statement characterizes identification in this
design.

\begin{proposition}\label{prop:S3}
(i) Under $F=F_{\theta_0}$, $E[g_\varphi(Z_1,\theta_0)]=0$ for all
$\varphi\in\mathcal{C}_{\mathrm{ind}}$. (ii) Conversely, if $E[g_\varphi(Z_1,\theta_0)]=0$
for all $\varphi\in\mathcal{C}_{\mathrm{ind}}$ and $dG(x)>0$ for $x\in\mathcal{S}_X$,
then $F=F_{\theta_0}$ on $\mathcal{S}_X$.
\end{proposition}

\subsection{MLE, orthogonal influence function, and sufficient conditions}

Under $\mathcal{H}_0$, the conditional likelihood of $\delta$ given $C=c$ is Bernoulli
with success probability $F_\theta(c)$. The loglikelihood contribution is
\[
\ell_\theta(d,c)=d\log F_\theta(c)+(1-d)\log(1-F_\theta(c)).
\]
Let $\dot F_\theta(c)=\partial F_\theta(c)/\partial\theta$ (a $p$-vector) and define
\begin{equation}
s_\theta(c)=\frac{\dot F_\theta(c)}{F_\theta(c)\big(1-F_\theta(c)\big)}.
\label{eq:CS_stheta}
\end{equation}
Then the score for $\theta$ is
\begin{equation}
l_\theta(d,c)= s_\theta(c)\big(d-F_\theta(c)\big).
\label{eq:CS_score_theta}
\end{equation}
Let $\hat\theta$ solve $E_n[l_{\hat\theta}(Z_1)]=0$.

Since $\partial g_\varphi(\delta,c;\theta)/\partial\theta'=-\varphi(c)\dot F_\theta(c)'$,
the orthogonal influence function for $E_n[g_\varphi(Z_1,\hat\theta)]$ takes the form
\begin{equation}
\Pi^\perp\varphi(z)
=
g_{\varphi}(z,\theta_0)
-
E\!\big[\varphi(C_1)\dot F_{\theta_0}(C_1)'\big]\,
I_{\theta_0}^{-1}l_{\theta_0}(z),
\label{eq:CS_pi}
\end{equation}
where $I_{\theta_0}=E[l_{\theta_0}(Z_1)l_{\theta_0}(Z_1)']$ is the Fisher information matrix. Sufficient conditions for Assumption A2 are given in the following lemma. To verify Assumption B2(ii), see details on MLE with current status data in \cite{KoulYi2006}.
In particular, Corollary 3.1 in \cite{KoulYi2006} gives conditions for the
asymptotic normality of the MLE in the case of the exponential scale family.
Discussion and references on MLE for general interval censored data are
given in \cite{Sun2006}, Sections 2.3 and 2.5.

\begin{lemma}\label{lem:E6}
A sufficient condition for Assumption~A2 in the current status setting is that, for a compact
$\Theta_0\subset\Theta$ containing $\theta_0$, $\theta\mapsto \dot F_{\theta}(C_1)$ is continuous a.s.\ and
\[
E\!\left[K(C_1,C_1)\sup_{\theta\in\Theta_0}\|\dot F_{\theta}(C_1)\|^2\right]<\infty,
\]
together with the standard MLE regularity in Assumption~B2(ii) of the main manuscript.
\end{lemma}

\begin{lemma}\label{lem:E7}
The mapping $\Pi^\perp$ in \eqref{eq:CS_pi} is norm reducing (under the null), so
$E[K_g^\perp(Z_1,Z_1)]<\infty$ holds under $E[K(C_1,C_1)\delta_1]<\infty$.
\end{lemma}

\subsection{Computation notes}

For the representation in Assumption~C1 (main manuscript), note that for each $(d,c)$,
\[
g_\varphi(d,c;\theta)=\langle \varphi,\ a_\theta(d,c)\rangle_K,
\qquad
a_\theta(d,c)=\big(d-F_\theta(c)\big)K_c,
\]
since $\langle\varphi,K_c\rangle_K=\varphi(c)$. The quadratic-form representation in
Lemma~\ref{Lemma:QuadraticForm} (main manuscript) then applies directly once the
orthogonalization term based on $l_{\hat\theta}$ is added.

\subsection{Relationship to existing tests}

For current status data (interval censoring case 1), our test is related to the Cram\'{e}r-von Mises test in \cite{KoulYi2006}. To be specific, in the simple null case the test statistic of \cite{KoulYi2006} equals our $n\hat Q_K$ when constructed from the particular class of functions $\varphi_x(c)=I(c\leq x)\{F_0(c)(1-F_0(c))\}^{-1/2}$. The approach for the asymptotic analysis in \cite{KoulYi2006} is, however, different, relying on weak convergence on compacta using \cite{Stute1998}. Again, our RKHS framework allows for convergence results on arbitrary supports. For interval censoring case 2,~\cite{Ren2003,Omidi2021} proposed NPMLE-based tests with slow convergence. We are not aware of other formal goodness-of-fit tests for interval-censored data.

%%%%%%%%%%%%%%%%%%%%%%%%%%%%%%%%%%%%%%%%%%%%%%%%%%%%%%%%%%%%%%%%%%%%%%%%%%%%%%
\section{Additional numerical results: random sampling (complete data)}
\label{SM:sim_random}

This section reports Monte Carlo evidence for the running example of random sampling
from $F$ with complete observations. These results complement the numerical study for
random double truncation in Section~\ref{sec:num} of the main manuscript.

\subsection{Simulation study: random sampling}
\label{sub:sim_random}

We consider random sampling with complete observations. The composite null hypothesis is
that $X$ is exponentially distributed,
\[
F_\theta(x)=1-\exp(-\theta x), \qquad x>0,
\]
where the rate parameter $\theta$ is unknown. We simulate $X$ from a Gamma distribution
with shape parameter $\theta$ and scale parameter $1$, with
$\theta \in \{0.5,0.8,1,1.2,1.5\}$. The null corresponds to $\theta=1$, while the other
values yield non-exponential alternatives. We use 1000 Monte Carlo trials under
$\mathcal{H}_0$ and 500 trials under $\mathcal{H}_1$, and sample sizes
$n\in\{50,100,200\}$.

The parameter $\theta$ is estimated by maximum likelihood under the null, i.e.,
$\hat\theta$ is the inverse of the sample mean. We report rejection proportions of the
statistics $n\hat Q_K^{(j)}$, $j=1,2$, associated with the score in
Remark~\ref{remark:alternative} and the score in \eqref{example1}, respectively. We use
the RKHS induced by the indicator class $\varphi_x(\tilde x)=\1\{\tilde x\le x\}$, and
compute critical values via a multiplier bootstrap with $B=499$ resamples.

For implementation we follow the generic computation in Section~\ref{Computation} of the
main manuscript: we form the induced-kernel Gram matrix $\widehat{\mathbb K}_g^\perp$ as follows. The $n \times n$ matrix $\mathbb{\hat{K}}_g^{\perp}$ with $ij$-th element $\hat{K}_g^{\perp}(Z_i,Z_j)$ is
\[
\mathbb{\hat{K}}_g^{\perp} = \hat{\Pi}_{l_{\hat{\theta}}}^{\perp} \mathbb{\hat{K}}_g \hat{\Pi}_{l_{\hat{\theta}}}^{\perp},
\]
where $\hat{\Pi}_{l_{\hat{\theta}}}^{\perp} = I_n - \mathbf{l}_{\hat{\theta}} (\mathbf{l}_{\hat{\theta}}' \mathbf{l}_{\hat{\theta}})^{-1} \mathbf{l}_{\hat{\theta}}$, $I_n$ is the $n\times n$ identity matrix, $\mathbf{l}_{\hat{\theta}}$ is the $n \times p$ matrix of scores with $i$-th row $l_{\hat{\theta}}'(Z_i)$, and $\mathbb{\hat{K}}_g$ is a $n\times n$ matrix. For $n\hat Q_K^{(2)}$, the $ij$-th element of $\mathbb{\hat{K}}_g$ is given by
\begin{align*}
\hat{K}_g(z_i,z_j) &= K(z_i,z_j) - \int_{\mathcal{S}_{Z_1}} K(z_i,z) f_{\hat{\theta}}(z) dz - \int_{\mathcal{S}_{Z_1}} K(z_j,z) f_{\hat{\theta}}(z) dz \\
&\quad + \int_{\mathcal{S}_{Z_1}} \int_{\mathcal{S}_{Z_1}} K(z,\tilde{z}) f_{\hat{\theta}}(z) f_{\hat{\theta}}(\tilde{z}) dz d\tilde{z}
\end{align*} so $n\hat Q_K^{(2)}=n^{-1}\sum_{i=1}^n\sum_{j=1}^n \hat K_g^\perp(X_i,X_j)$, with the
multiplier bootstrap based on the same $\hat K_g^\perp$. An analogous expression holds for $n\hat Q_K^{(1)}$, which is based on a different $\hat K_g$.

Table~\ref{tab:randomsampling} reports rejection proportions at significance level
$\alpha=0.05$. The nominal level is well approximated by the multiplier bootstrap. Under
alternatives, power increases with sample size and with the magnitude of the departure
from the null (as measured by $|\theta-1|$). The table also shows that no test uniformly
dominates the other: for $\theta<1$, $n\hat Q_K^{(2)}$ tends to have larger power, whereas
for $\theta>1$, $n\hat Q_K^{(1)}$ is more powerful.

\begin{table}[tbp]
\begin{center}
{\small \addtolength{\tabcolsep}{-2pt}
\begin{tabular}{cccccccc}
&  & $n\hat Q_K^{(1)}$ &  &  &  & $n\hat Q_K^{(2)}$ &  \\
$\theta$ & $n=50$ & $n=100$ & $n=200$ &  & $n=50$ & $n=100$ & $n=200$ \\
\hline
&  &  &  &  &  &  &  \\
$0.5$ & .812 & .982 & 1.00 &  & .836 & .986 & 1.00 \\
$0.8$ & .126 & .232 & .464 &  & .130 & .254 & .490 \\
$1$ & .054 & .058 & .063 &  & .061 & .057 & .060 \\
$1.2$ & .188 & .256 & .444 &  & .168 & .242 & .420 \\
$1.5$ & .552 & .842 & .984 &  & .528 & .818 & .982 \\
&  &  &  &  &  &  &  \\ \hline
\end{tabular}
}
\end{center}
\caption{Rejection proportions of the test statistics $n\hat Q_K^{(1)}$ and
$n\hat Q_K^{(2)}$ along the Monte Carlo trials at level $\protect\alpha=0.05$.
The null hypothesis corresponds to $\protect\theta=1$. Random sampling.}
\label{tab:randomsampling}
\end{table}

%%%%%%%%%%%%%%%%%%%%%%%%%%%%%%%%%%%%%%%%%%%%%%%%%%%%%%%%%%%%%%%%%%%%%%%%%%%%%%
\section{A Uniform Law of Large Numbers}
\label{ULLN}

Let \(\mathcal H_K\) be the RKHS with kernel \(K\), \(\langle\cdot,\cdot\rangle_K\) its inner product, and \(\mathcal B_K=\{\varphi\in\mathcal H_K:\|\varphi\|_K\le 1\}\) its unit ball. Let \(\Theta\subset\mathbb R^p\), and let $\Theta_0$ be a compact subset of $\Theta$. For a measurable map
\[
s:\mathcal S_{Z_1}\times \Theta_0 \times \mathcal H_K \to \mathbb R,\qquad
(z,\theta,\varphi)\mapsto s_\varphi(z,\theta),
\]
we use the following notation: whenever \(\varphi\mapsto s_\varphi(z,\theta)\) is a bounded linear functional on \(\mathcal H_K\), we denote by \(r_s(z,\theta,\cdot)\in\mathcal H_K\) its Riesz representer, so that $s_\varphi(z,\theta)=\langle \varphi,\,r_s(z,\theta,\cdot)\rangle_K$ for all $\varphi\in\mathcal H_K$. Define the induced kernel
\[
K_s(z_1,z_2,\theta):=\langle r_s(z_1,\theta,\cdot),\,r_s(z_2,\theta,\cdot)\rangle_K,
\qquad
K_s(z,\theta):=K_s(z,z,\theta)=\|r_s(z,\theta,\cdot)\|_K^2.
\]

\noindent Define the empirical mean process \(S_n(\varphi,\theta):=E_n[s_\varphi(Z_1,\theta)]\), where \(E_n\) denotes the sample average. We work under the following representer and regularity conditions.

\bigskip

\noindent \textbf{Assumption B1.} For each $z_1,z_2\in \mathcal{S}_{Z_{1}}$ and $\theta \in \Theta _{0}$, with $\Theta_0$ be a compact subset of $\Theta \subset \mathbb R^p$: (i) $\varphi \mapsto s_{\varphi}(z_1,\theta)$ belongs to $\mathcal{H}_{K}^{\ast }$; (ii) $\theta \mapsto K_{s}(z_1,z_2,\theta )$ is continuous in $\Theta _{0}$ a.s., with
    \begin{equation*}
    E\Big[ \sup_{\theta \in \Theta _{0}} |K_{s}(Z_{1},\theta)| \Big] < \infty.
    \end{equation*}

\begin{lemma}[Uniform Law of Large Numbers]
\label{lem:B1}
Under Assumption B1,
\begin{align*}
\sup_{\varphi \in \mathcal{B}_{K},\theta \in \Theta _{0}} \big| S_{n}(\varphi ,\theta )-E[s_{\varphi }(Z_{1},\theta )] \big| &= o_{P}(1), \\
\sup_{\varphi \in \mathcal{B}_{K}} \big| S_{n}(\varphi ,\theta _{n})-E[s_{\varphi }(Z_{1},\theta _{0})] \big| &= o_{P}(1),
\end{align*}
for any consistent estimator $\theta_n$ of $\theta_0$.
\end{lemma}

\noindent \textsc{Proof of Lemma~\ref{lem:B1}:}
By Assumption B1(i) and the Riesz Representation theorem, there exists $r_{s}(z,\theta ,\cdot )\in \mathcal{H}_{K}$ such that
\[
s_{\varphi }(z,\theta )=\left\langle \varphi ,r_{s}(z,\theta ,\cdot )\right\rangle _{K}, \quad r_{s}(z,\theta ,z_{1})=s_{K_{z_{1}}}(z,\theta ).
\]
Then, using standard Hilbert space arguments,
\begin{align*}
&\sup_{\varphi \in \mathcal{B}_{K},\theta \in \Theta _{0}}\left\vert S_{n}(\varphi ,\theta )-E[s_{\varphi }(Z_{1},\theta )]\right\vert ^{2} \\
&\leq \sup_{\theta \in \Theta _{0}}\left\Vert E_{n}[r_{s}(Z_{1},\theta ,\cdot )]-E[r_{s}(Z_{1},\theta ,\cdot )]\right\Vert _{K}^{2} \\
&\leq \sup_{\theta \in \Theta _{0}} \big| \frac{1}{n^{2}}\sum_{i=1}^{n}\sum_{j=1}^{n}K_{s}(Z_{i},Z_{j},\theta)-E[K_{s}(Z_{i},Z_{j},\theta)]\big|  \\
&= o_{P}(1),
\end{align*}
where the last equality follows from \cite{NolanPollard1987}, the continuity of Assumption B1, together with
\[
E\Big[ \sup_{\theta \in \Theta _{0}} | K_s(Z_{i},Z_{j},\theta) | \Big] \le E[ \sup_{\theta \in \Theta_0} |K_s(Z_{i},\theta)| ]  < \infty.
\]
The second part follows from the ordinary ULLN and the continuous mapping theorem. $\square$

%%%%%%%%%%%%%%%%%%%%%%%%%%%%%%%%%%%%%%%%%%%%%%%%%%%%%%%%%%%%%%%%%%%%%%%%%%%%%%
\section{Proofs of general results}
\label{GeneralProofs}

\noindent \textsc{Proof of Theorem~\ref{Thm:GeneralNull}:}
By Assumptions~A1--A3 and Slutsky's theorem, $\sqrt{n}\,\hat{\Delta}_n^\perp \Rightarrow S_\infty \quad \text{in } \mathcal{H}_K^\ast.$ Then, the result follows from the continuous mapping theorem and $n\hat{Q}_K = \|\sqrt{n}\,\hat{\Delta}_n^\perp\|^2$. $\square$\medskip

\noindent \textsc{Proof of Lemma~\ref{Lemma:B2}:}
Note
\begin{eqnarray*}
E_{n}[\hat{\Pi}^{\perp }\varphi (Z_{1})-\Pi ^{\perp }\varphi (Z_{1})]
&=&E_{n}[g_{\varphi }(Z_{1},\hat{\theta})-g_{\varphi }(Z_{1},\theta _{0})] \\
&&+B(\varphi ,\theta _{0})\left( \hat{I}^{-1}E_{n}[\hat{l}(Z_{1})]-I_{\theta
_{0}}^{-1}E_{n}[l_{\theta _{0}}(Z_{1})]\right) \\
&&+\left( B_{n}(\varphi ,\hat{\theta})-B(\varphi ,\theta _{0})\right) \hat{I}%
^{-1}E_{n}[\hat{l}(Z_{1})] \\
&=&:I_{1n}(\varphi ,\hat{\theta})+I_{2n}(\varphi ,\hat{\theta}%
)+I_{3n}(\varphi ,\hat{\theta}).
\end{eqnarray*}%
By a mean value theorem
\begin{equation*}
I_{1n}(\varphi ,\hat{\theta})=B^{\partial}_{n}(\varphi ,\bar{\theta})(\hat{\theta}%
-\theta _{0})
\end{equation*}%
where $\bar{\theta}$ is an intermediate value such that $\left\vert \bar{%
\theta}-\theta _{0}\right\vert \leq \left\vert \hat{\theta}-\theta
_{0}\right\vert $ a.s. For simplicity of exposition, we will use the same
notation $\bar{\theta}$ for any intermediate value (though, of course, it
will change from equation to equation). Likewise, using B2(ii)$,$ uniformly
in $\varphi \in \mathcal{B}_{K},$
\begin{equation*}
I_{2n}(\varphi ,\hat{\theta})=-B(\varphi ,\theta _{0})(\hat{\theta}-\theta
_{0})+o_{P}(n^{-1/2}).
\end{equation*}%
Also, from Assumption B2, $\hat{I}^{-1}E_{n}[\hat{l}(Z_{1})]=O_{P}(n^{-1/2})$
and $I_{3n}(\varphi ,\hat{\theta})=o_{P}(n^{-1/2}),$ uniformly in $\varphi
\in \mathcal{B}_{K}$. Therefore, conclude that, uniformly in $\varphi \in
\mathcal{B}_{K},$%
\begin{eqnarray*}
E_{n}[\hat{\Pi}^{\perp }\varphi (Z_{1})-\Pi ^{\perp }\varphi (Z_{1})]
&=&\left( B^{\partial}_{n}(\varphi ,\bar{\theta})-B(\varphi ,\theta _{0})\right) (\hat{%
\theta}-\theta _{0})+o_{P}(n^{-1/2}) \\
&=&o_{P}(n^{-1/2}).
\end{eqnarray*}%
$\square$ \medskip

\noindent \textsc{Proof of Lemma~\ref{Lemma:B3}:}
It suffices to apply Theorem 1 in \cite{FernandezRivera2022}, or a slightly simplified
version of it in Theorem 4.1 in \cite{Escanciano2024}, with $%
S_{n}(a)=E_{n}[\Pi ^{\perp }a(Z_{1})],$ $a\in \mathcal{H}_{K}.$
The proof of (ii) follows from the norm-reducing property, since $K_{g}^{\perp}(z,z)\leq C K_{g}(z,z)$.
The continuity follows from Proposition 2.5 in \cite{FerreiraMenegatto2013}. $\square $\medskip

\noindent \textsc{Proof of Lemma~\ref{Lemma:QuadraticForm}:}
By assumption there exists a measurable $\mathcal{H}_K$ valued map $\hat h$ such that $(\hat \Pi^\perp \varphi)(z)=\langle \varphi, \hat h(z)\rangle_{K}$ for all $\varphi \in \mathcal{H}_K$ and $z$. Hence
\[
\hat{\Delta}_n^\perp(\varphi)
= E_n\!\big[(\hat\Pi^\perp\varphi)(Z_1)\big]
= E_n\!\big[\langle \varphi,\hat h(Z_1)\rangle_K\big]
= \Big\langle \varphi,\ E_n[\hat h(Z_1)]\Big\rangle_K .
\]

\noindent Define $\hat v_n := E_n[\hat h(Z_1)]$. By Cauchy-Schwarz,
\[
\hat Q_K
=\sup_{\|\varphi\|_K\le 1}\langle \varphi,\hat v_n\rangle_K^2
=\|\hat v_n\|_K^2 .
\]

\noindent Expanding the squared norm,
\[
\|\hat v_n\|_K^2
=\left\langle \frac{1}{n}\sum_{i=1}^n \hat h(Z_i),\ \frac{1}{n}\sum_{j=1}^n \hat h(Z_j)\right\rangle_K
=\frac{1}{n^2}\sum_{i=1}^n\sum_{j=1}^n \langle \hat h(Z_i),\hat h(Z_j)\rangle_K .
\]
With $\hat K_g^\perp(z,z'):=\langle \hat h(z),\hat h(z')\rangle_K$, this yields
\(
\hat Q_K=\frac{1}{n^2}\sum_{i,j}\hat K_g^\perp(Z_i,Z_j).
\)
The proof for the bootstrap version is analogue.$\square$\medskip

%\noindent \textsc{Proof of Corollary~\ref{cor:P1}:}
%Apply the uniform expansion in (\ref{P1}), Lemma~\ref{Lemma:B2} with $s_\varphi(z_1,\theta) = \Pi^\perp \varphi(Z_1)$ (independent of $\theta$), and the continuous mapping theorem. $\square$\medskip

\noindent \textsc{Proof of Theorem~\ref{thm:P1}:}
Define $h_{\theta}(z)=a_{\theta}(z)+b_{\theta}I_{\theta}^{-1} l_{\theta}(z)$. Then, note that $E\left[h_{\theta_0}(Z_1)\right]=E\left[a_{\theta_0}(Z_1)\right]$ by $E[l_{\theta_0}(Z_1)] = 0$. Conclude by standard RKHS theory, since $E[K_{g}^{\perp }(Z_{1},Z_{2})]= \|E\left[h_{\theta_0}(Z_1)\right]\|^2_K$ and $\|E\left[a_{\theta_0}(Z_1)\right]\|^2_K=d_{K}(F,F_{0})>0$.
$\square$

\section{Proofs of the examples}
\label{ExamplesProofs}

\subsection{Proofs for left-truncation with right-censoring} \mbox{}\par
\medskip

\noindent \textsc{Proof of Proposition \ref{prop:S1}}: Part (i) immediately follows
from the independence between $(U,C)$ and $X$ and condition $P(U\leq C)=1$. Certainly, with $\alpha =P(U\leq X)$, it holds

\begin{equation*}
E[\varphi(Y_1)\delta_1]=\alpha^{-1}E[\varphi(X)I(U\leq X\leq
C)]=\alpha^{-1}E[\varphi(X)M(X)].
\end{equation*}

\noindent Now, by taking a target $X$ independent of $(U_1,Y_1)$, under $%
F=F_{\theta_0}$ one may write

\begin{equation*}
E_{\theta_0}[\psi_{\theta_0}^\varphi(Z_1)]=E\left[E \left(\frac{%
\varphi(X)I(U_1\leq X\leq Y_1)}{1-F(X-)}| U_1,Y_1 \right) \right],
\end{equation*}

\begin{equation*}
=E\left[E \left(\frac{\varphi(X)I(U_1\leq X\leq Y_1)}{1-F(X-)}| X \right) %
\right]=\alpha^{-1}E[\varphi(X)M(X)],
\end{equation*}
\noindent where the last equality follows by noting $P(U_1\leq x\leq
Y_1)=\alpha^{-1}M(x)(1-F(x-))$, $x \in \mathcal{S}_X$. This completes the
proof of (i). Conversely, for (ii) one starts from equality $E[g_{\varphi }(Z_{1},\theta _{0})]=0$ for each $\varphi =\varphi _{x}$ in $\mathcal{C}_{\mathrm{ind}}$. Now,
since $E[\varphi (Y_{1})\delta _{1}]=\alpha ^{-1}E[\varphi (X)M(X)]$ and $%
E[\psi _{\theta _{0}}^{\varphi }(Z_{1})]=\alpha ^{-1}E_{\theta _{0}}[\varphi
(X)M(X)(1-F(X-))/(1-F_{\theta _{0}}(X-))]$, and because $M(x)>0$ for $x\in
\mathcal{S}_{X}$, by using the definition of $\mathcal{C}_{\mathrm{ind}}$ one may conclude $(1-F(x-))^{-1}dF(x)=(1-F_{\theta
_{0}}(x-))^{-1}dF_{\theta _{0}}(x)$ along the support of $X$. Since the
hazard function characterizes the cdf, the result follows. $\square$\medskip

\noindent \textsc{Proof of Lemma \ref{lem:E1}}: The result follows from Lemma \ref{lem:B1} after noting that, for the left-truncated and right-censored scenario, the function $K_s(z_1,z_2,\theta)$ with $z_i=(y_i,u_i,d_i)$, $i=1,2$, reduces to $K_s(z_1,z_2,\theta)=h_\theta(y_1)'d_1K(y_1,y_2)h_\theta(y_2)d_2.$ $\square$\medskip

\noindent \textsc{Proof of Lemma \ref{lem:E2}}: According to (\ref{eq:LTRC_pi}), we have $\Pi^\perp(z)=g_\varphi(z,\theta)+B(\varphi,\theta)I_\theta^{-1}l_\theta(z)$. Now, since the score $g_\varphi$ is zero-mean and $B(\varphi,\theta)I_\theta^{-1}B(\varphi,\theta)'\geq 0$, it holds $E[(\Pi^\perp\varphi(Z_1))^2]\leq E[(g_\varphi(Z_1,\theta))^2]$. Then, it suffices to prove $E[(\varphi
(Y_{1})\delta _{1}-\psi _{\theta }^{\varphi }(Z_{1}))^{2}]\leq E[\varphi
(Y_{1})^{2}\delta _{1}]$. Now, with $a=\varphi (Y_{1})\delta _{1}$ and $b=\psi^\varphi_\theta(Z_1)$, it holds

\begin{equation*}
E[(\varphi(Y_1)\delta_1-\psi^\varphi_\theta(Z_1))^2]=E[\varphi(Y_1)^2\delta_1]
+ E[b^2]-2E[ab].
\end{equation*}

\noindent Note that

\begin{equation*}
E[b^2]=\int \int \frac{\varphi(x)\varphi(\tilde x)}{(1-F_\theta(x-))(1-F_%
\theta(\tilde x-))}E[I(U_1\leq x\wedge \tilde x,x \vee \tilde x\leq
Y_1)]dF_\theta(x)dF_\theta(\tilde x)
\end{equation*}

\begin{equation*}
\leq2\int \int \frac{\varphi(x)\varphi(\tilde x)E[I(U_1\leq x,\tilde x\leq
Y_1)]}{(1-F_\theta(x-))(1-F_\theta(\tilde x-))}I(x\leq \tilde
x)dF_\theta(x)dF_\theta(\tilde x);
\end{equation*}

\noindent the inequality is indeed an equality when $F_\theta$ is
continuous. Also,

\begin{equation*}
-2E[ab]=-2\int \frac{\varphi(x)}{(1-F_\theta(x-))}E[\varphi(Y_1)\delta_1
I(U_1\leq x\leq Y_1)]dF_\theta(x).
\end{equation*}

\noindent Since

%\begin{equation*}

%\end{equation*}

$$E[\varphi(Y_1)\delta_1 I(U_1\leq x\leq Y_1)]=\alpha^{-1}E[\varphi(X)I(X\leq
C)I(U\leq x\leq X)]$$

$$=\alpha^{-1}\int \varphi(\tilde x)E[I(U\leq x,\tilde x\leq C)]I(x\leq \tilde
x)dF_\theta(\tilde x)$$

\noindent and, for $x \leq \tilde x$,

\begin{equation*}
E[I(U_1\leq x,\tilde x\leq Y_1)]=\alpha^{-1}E[I(U\leq x,\tilde x\leq
C)](1-F_\theta(x-)),
\end{equation*}

\noindent we get $E[b^{2}]-2E[ab]\leq 0$, with an equality for continuous $%
F_{\theta }$. $\square$

%\noindent \textsc{Proof of Lemma \ref{lem:E2}}: For a given $G$, let $\alpha _{\tau }$

\subsection{Proofs for doubly truncated data}\mbox{}\par
\medskip

\noindent \textsc{Proof of Proposition \ref{prop:S2}}: Part (i) immediately follows
from the independence between $(U,V)$ and $X$. Certainly, it holds

\begin{equation*}
E[\varphi(X_1)|U_1,V_1]=\displaystyle \int \varphi(\tilde x)\frac{I(U_1\leq
\tilde x\leq V_1)}{F(V_1)-F(U_1-)}dF(\tilde x).
\end{equation*}

\noindent Now, since $\psi_{\theta_0}^\varphi(Z_1)$ is free of $X_1$, one
gets

\begin{equation*}
E_{\theta_0}[\psi_{\theta_0}^\varphi(Z_1)|U_1,V_1]=\psi_{\theta_0}^%
\varphi(Z_1)=\displaystyle \int \varphi(\tilde x)\frac{I(U_1\leq \tilde
x\leq V_1)}{F_{0}(V_1)-F_{0}(U_1-)}dF_{0}(\tilde x).
\end{equation*}

\noindent Thus, under $F=F_{\theta_0}$ it holds $E[\varphi(X_1)|U_1,V_1]=
E_{\theta_0}[\psi_{\theta_0}^\varphi(Z_1)|U_1,V_1]$ and (i) follows. Conversely, assume $E[g_{\varphi }(Z_{1},\theta _{0})]=0$ for each $\varphi=\varphi_x$ in $\mathcal{C}_{\mathrm{ind}}$. Since the
joint density of $(U_1,V_1)$ is given by $\alpha^{-1}(F(v)-F(u-))dG(u,v)$
where $\alpha=P(U\leq X \leq V)$, one gets

%\begin{equation*}
%E[\varphi(X_1)]=\alpha^{-1}\int \varphi(\tilde x)M(\tilde x)dF(\tilde x),
%\end{equation*}
%
%\noindent and, similarly,
%
%\begin{equation*}
%E[\psi _{\theta _{0}}^{\varphi }(Z_{1})]=\alpha ^{-1}\int \varphi (\tilde{x}%
%)M_{\theta _{0}}(\tilde{x})dF_{\theta _{0}}(\tilde{x}),
%\end{equation*}

\[
E[\varphi(X_1)]=\alpha^{-1}\!\int \varphi(\tilde x)M(\tilde x)\,dF(\tilde x),\qquad
E[\psi_{\theta_0}^{\varphi}(Z_1)]=\alpha^{-1}\!\int \varphi(\tilde{x})M_{\theta_0}(\tilde{x})\,dF_{\theta_0}(\tilde{x}),
\]

\noindent where

\begin{equation*}
M_{\theta}(x)=E_\theta\left[\frac{F(V)-F(U-)}{F_\theta(V)-F_\theta(U-)}I(U\leq x\leq
V)\right].
\end{equation*}

\noindent Since the last two equalities hold
along $\mathcal{C}_{\mathrm{ind}}$, (ii) follows from (\ref{key:DT}).$\square$\medskip

\noindent \textsc{Proof of Lemma \ref{lem:E3}}: The result follows from Lemma \ref{lem:B1} after noting that, for the doubly truncated setting, the function $K_s(z_1,z_2,\theta)$ with $z_i=(x_i,u_i,v_i)$, $i=1,2$, is given by $K_s(z_1,z_2,\theta)=K(x_1,x_2)l_\theta(z_1)'l_\theta(y_2).$ $\square$

\bigskip

\noindent \textsc{Proof of Lemma \ref{lem:E4}}: The proof is analogous to that of Lemma \ref{lem:E2}. Due to (\ref{eq:DT_pi}), it suffices to prove $E[(\varphi
(X_{1})-\psi _{\theta }^{\varphi }(Z_{1}))^{2}]\leq E[\varphi (X_{1})^{2}]$. Since $E[\varphi (X_{1})|U_{1},V_{1}]=E[\psi _{\theta }^{\varphi
}(Z_{1})|U_{1},V_{1}]$ and since $\psi _{\theta }^{\varphi }(Z_{1})$ is free
of $X_{1}$, it holds

\begin{equation*}
E[(\varphi(X_1)-\psi_\theta^\varphi
(Z_1))^2|U_1,V_1]=V[\varphi(X_1)|U_1,V_1]\leq E[\varphi(X_1)^2|U_1,V_1],
\end{equation*}

\noindent where $V[\cdot ]$ stands for the variance operator. Take the
expectation with respect to $(U_{1},V_{1})$ at both sides of the inequality
to conclude. $\square$

\subsection{Proofs for current status data}\mbox{}\par
\medskip

\noindent \textsc{Proof of Proposition \ref{prop:S3}}: Part (i) immediately follows
from the independence between $C$ and $X$. Conversely, for (ii) one starts
from the equality $E[g_{\varphi }(Z_{1},\theta _{0})]=0$ for each $\varphi=\varphi_x$ in $\mathcal{C}_{\mathrm{ind}}$, and concludes that $F(x)dG(x)=F_{0}(x)dG(x)$, $x \in \mathcal{S}_X$, by the fundamental theorem of calculus. Using condition $dG(x)>0$ for $x\in \mathcal{S}_X$, the proof is complete. $\square$

\bigskip

\noindent \textsc{Proof of Lemma \ref{lem:E6}}: The result follows from Lemma \ref{lem:B1} after noting that, with current status data, the function $K_s(z_1,z_2,\theta)$ with $z_i=(d_i,c_i)$, $i=1,2$, corresponds to $K_s(z_1,z_2,\theta)=K(c_1,c_2)s_\theta(c_1)'s_\theta(c_2).$ $\square$\medskip

\noindent \textsc{Proof of Lemma \ref{lem:E7}}: The proof is analogous to that of Lemma \ref{lem:E2}, where here equation (\ref{eq:CS_pi}) is the key. Note that

\begin{equation*}
E[(\varphi (C_{1})\delta _{1}-\psi _{\theta }^{\varphi}(Z_{1}))^{2}]
=E[\varphi (C_{1})^{2}(\delta _{1}-F_{\theta }(C_{1}))^{2}],
\end{equation*}

\noindent and this is upper bounded by $E[\varphi (C_{1})^{2}\delta _{1}]$ because

\begin{equation*}
E[\varphi (C_{1})^{2}(-2\delta _{1}F_{\theta }(C_{1})+F_{\theta }(C_{1})^{2})]
=-E[\varphi (C_{1})^{2}F_{\theta }(C_{1})^{2}]\leq 0.
\end{equation*}

\noindent This concludes the proof. $\square$

\fi
\end{document}